\documentclass[journal]{IEEEtran}
\usepackage[english]{babel}
\usepackage{amsmath,amsfonts}
\usepackage{amsthm}
\usepackage{mathtools}
\usepackage{algorithm}
\usepackage[noend]{algpseudocode}
\usepackage{array}
\usepackage[caption=false,font=normalsize,labelfont=sf,textfont=sf]{subfig}
\usepackage{textcomp}
\usepackage{stfloats}
\usepackage{url}
\usepackage{verbatim}
\usepackage{graphicx}
\usepackage{cite}
\usepackage{comment}
\usepackage[dvipsnames]{xcolor}
\usepackage{hyperref}
\usepackage{booktabs}
\usepackage{siunitx}
\usepackage{stackrel}
\usepackage{enumitem}
\usepackage{multirow}
\usepackage{subfiles}

\usepackage{pifont}
\newcommand{\cmark}{\ding{51}}%
\newcommand{\xmark}{\ding{55}}%

\usepackage{tablefootnote}
\DeclareMathOperator{\R}{{\mathbb{R}}}

\renewcommand{\vec}[1]{\mathbf{#1}}
\newcommand{\mat}[1]{\mathbf{#1}}

\newcommand{\param}{\boldsymbol{\theta}} 

\DeclareMathOperator*{\argmin}{arg\,min} 
\DeclareMathOperator{\prox}{prox} 

\renewcommand{\vec}[1]{\mathbf{#1}}

\newcommand{\N}{\mathbb{N}}

\DeclareMathOperator{\Lip}{Lip}

\def\V#1{{\boldsymbol{#1}}}  

\newtheorem{theorem}{Theorem}[section]

\newtheorem{remark}[theorem]{Remark}

\newtheorem{corollary}[theorem]{Corollary}
\newtheorem{proposition}[theorem]{Proposition}

\newcommand{\cmarkc}{{\color{ForestGreen}\cmark}}
\newcommand{\xmarkc}{{\color{Red}\xmark}}
\title{A Neural-Network-Based Convex Regularizer\\ for Inverse Problems}
\date{\today}
\author{Alexis Goujon, \and Sebastian Neumayer, \and Pakshal Bohra, \and Stanislas Ducotterd, \and and Michael Unser}

\begin{document}
\maketitle
\renewcommand*{\thefootnote}{\fnsymbol{footnote}}
\footnotetext[1]{The authors are with the Biomedical Imaging Group, \'Ecole polytechnique f\'ed\'erale de Lausanne (EPFL),
	Station 17, CH-1015 Lausanne, {\text \{forename.name\}@epfl.ch}. }
\renewcommand*{\thefootnote}{\arabic{footnote}}

\begin{abstract}
    The emergence of deep-learning-based methods to solve image-reconstruction problems has enabled a significant increase in quality.
     Unfortunately, these new methods often lack reliability and explainability, and there is a growing interest to address these shortcomings while retaining the boost in performance.
    In this work, we tackle this issue by revisiting regularizers that are the sum of convex-ridge functions.
    The gradient of such regularizers is parameterized by a neural network that has a single hidden layer with increasing and learnable activation functions.
    This neural network is trained within a few minutes as a multistep Gaussian denoiser.
    The numerical experiments for denoising, CT, and MRI reconstruction show improvements over methods that offer similar reliability guarantees.
\end{abstract}

\section{Introduction}
In natural science, it is common to indirectly probe an object of interest by collecting a series of linear measurements~\cite{ribes2008linear}.
After discretization, this can be formalized as
\begin{equation}
    \vec y = \vec H\vec x + \vec n,
    \label{eq:InvProb}
\end{equation}
where $\vec H \in \R^{m\times d}$ acts on the discrete representation $\vec x\in\R^d$ of the object and models the physics of the process.
The vector $\vec n \in\R^m$ accounts for additive noise in the measurements.
Given the measurement vector $\vec y \in\R^m$, the task is then to reconstruct $\vec x$.
Many medical-imaging applications fit into this class of inverse problems~\cite{MM2019}, including magnetic-resonance imaging (MRI) and X-ray
computed tomography (CT).

In addition to the presence of noise, which makes the reconstruction challenging for ill-conditioned $\vec H$, it is common to have only a few measurements ($m<d$), resulting in underdetermined problems.
In either case, \eqref{eq:InvProb} is ill-posed, and solving it poses serious challenges.
To overcome this issue, a reconstruction $\vec x^*$ is often computed as
\begin{equation}\label{eq:VarProb}
    \vec x^* \in \argmin \limits_{\vec x \in \R^d} \Vert \vec H\vec x - \vec y \Vert_2^2 + R(\vec x),
\end{equation}
where $R\colon \R^d\rightarrow \R$ is a convex regularizer that incorporates prior information about $\vec x$ to counteract the ill-posedness of \eqref{eq:InvProb}.
Popular choices are the Tikhonov \cite{tikhonov1963} or total-variation (TV) \cite{rudin1992nonlinear, donoho2006compressed, candes2008introduction} regularizers.
\subsection{Deep-Learning Methods}
Deep-learning-based methods have emerged in the past years for the inversion of \eqref{eq:InvProb} in a variety of applications; see \cite{AMOS2019,OngJalBar2020} for an overview.
Such approaches offer a significantly improved quality of reconstruction as compared to classical variational models of the form \eqref{eq:VarProb}.
Unfortunately, most of them are not well understood and lack stability guarantees \cite{antun2020instabilities,gottschling2020troublesome}.

For end-to-end approaches, a pre-trained model outputs a reconstruction directly from the measurements $\vec y$ or from a low-quality reconstruction \cite{jin2017deep,CZZLLZW2017,BJSBM2018, hyun2018deep, HN20}.
These approaches are often much faster than iterative solvers that compute \eqref{eq:VarProb}.
Their downside is that they offer no control of the data-consistency term $\Vert \vec H\vec x - \vec y \Vert_2$.
In addition, they are less universal since a model is specifically trained per $\vec H$ and per noise model.
End-to-end learning can also lead to serious stability issues~\cite{antun2020instabilities}.

A remedy for some of these issues is provided by the convolutional-neural-network (CNN) variants of the plug-and-play (PnP) framework~\cite{venkatakrishnan2013plug,chan2016plug,romano2017little,ryu2019plug}. 
The inspiration for these methods comes from the interpretation of the proximal operator
\begin{equation}
\label{eq:prox}
    \prox_R(\vec y) = \argmin\limits_{\vec{x}\in\R^d} \frac{1}{2}\|\vec{y} - \vec{x}\|_2^2 + R(\vec{x})
\end{equation}
used in many iterative algorithms for the computation of \eqref{eq:VarProb} as a denoiser.
The idea is to replace \eqref{eq:prox} with a more powerful CNN-based denoiser $\V D$.
However, $\V D$ is usually not a proper proximal operator, and the convergence of the PnP iterates is not guaranteed. It was shown in \cite{ryu2019plug} that, for an invertible $\mat H$, convergence can be ensured by constraining the Lipschitz constant of the residual operator $(\mathbf{Id}-\V D)$, where $\mathbf{Id}$ is the identity operator.
For a noninvertible $\mat H$, this constraint, however, does not suffice.
Instead, one can constrain $\V D$ to be an averaged operator which, unfortunately, degrades the performance \cite{bohra2021learning}.
Hence, in practice, one usually only constrains $(\mathbf{Id}-\V D)$, even if the framework is deployed for noninvertible $\mat H$ \cite{ryu2019plug,HHNPSS2019,HNS2021}.
While this results in good performances, it leaves a gap between theory and implementation.
Following a different route, one can also ensure convergence with relaxed algorithms \cite{gupta2018cnn,hurault2022gradient}.
There, $\V D$ is replaced with the relaxed version $\gamma \V D + (1-\gamma)\mathbf{Id}$, $\gamma\in(0,1]$. At each iteration, $\gamma$ is decreased if some condition is violated. Unfortunately, without particular constraints on $\V D$, the evolution of $\gamma$ is unpredictable.
Hence, the associated fixed-point equation for the reconstruction is unknown a priori, which reduces the reliability of the method.

Another data-driven approach arising from \eqref{eq:VarProb} is the learning of $R$ instead of $\prox_R$.
Pioneering work in this direction includes the \emph{fields of experts} \cite{RotBla2009,CheRan2014,EffKob2020}, where $R$ is parameterized by an interpretable and shallow model, namely, a sum of nonlinear one-dimensional functions composed with convolutional filters.
Some recent approaches rely on more sophisticated architectures with much deeper CNNs, such as with the adversarial regularization (AR) \cite{lunz2018adversarial,DufNeiEhr2021}, NETT \cite{LiSch2020}, and the total-deep-variation frameworks \cite{KobEff2020}, or with regularizers for which a proximal operator exists \cite{cohen2021has,hurault2022gradient,HurLec2022,fermanian2022learned}.
There exists a variety of strategies to learn $R$, including bilevel optimization \cite{CheRan2014}, unrolling \cite{KobEff2020,EffKob2020}, gradient-step denoising \cite{cohen2021has,hurault2022gradient}, and adversarial training \cite{lunz2018adversarial,DufNeiEhr2021}.
When $R$ is convex, a global minimizer of \eqref{eq:VarProb} can be found under mild assumptions.
As the relaxation of the convexity constraint usually boosts the performance \cite{CheRan2014,KKHP2017}, it is consequently the most popular approach.
Unfortunately, one can then expect convergence only to a critical point.
\subsection{Quest for Reliability}
In many sensitive applications such as medical imaging, there is a growing interest to improve the reliability and interpretability of the reconstruction methods.
The available frameworks used to learn a (pseudo) proximal operator or regularizer result in a variety of neural architectures that differ in the importance attributed to the following competing properties:
\begin{itemize}
    \item good reconstruction quality;
    \item independence on $\vec H$, noise model, and image domain;
    \item convergence guarantees and properties of the fixed points of the reconstruction algorithm;
    \item interpretability, which can include the existence of an explicit cost or a minimal understanding of what the regularizer is promoting.
\end{itemize}
 To foster the last two properties, one usually has to impose structural constraints on the learnt regularizer/proximal operator.
For instance, within the PnP framework, there have been some recent efforts to improve the expressivity of averaged denoisers, either with strict Lipschitz constraints on the model, \cite{bohra2021learning,NC2022} or with regularization of its Lipchitz constant during training \cite{PRTW2021,HurLec2022} which, in turn, improves the convergence properties of the reconstruction algorithm.
In the same vein, the authors of \cite{MukDit2021,MukSch2021} proposed to learn a convex $R$ parameterized by a deep input convex neural network (ICNN)\cite{AXK2016} and to train it within an adversarial framework as in \cite{lunz2018adversarial}.

In the present work, we prioritize the reliability and interpretability of the method.
Thus, we revisit the family of learnable convex-ridge regularizers \cite{RotBla2009,CheRan2014,KKHP2017,EffKob2020,nguyen2017learning}
\begin{equation}
\label{eq:convexridgereg}
    R\colon\vec{x}\mapsto \sum_{i}\psi_i(\vec{w}_i^T \vec{x}),
\end{equation}
where the profile functions $\psi_i\colon \R\rightarrow\R$ are convex, and $\vec{w}_i\in\R^d$ are learnable weights.
A popular way to learn $R$ is to solve a non-convex bilevel optimization task \cite{peyre2011learning, chen2012learning} for a given inverse problem.
It was reported in \cite{CheRan2014} that these learnt regularizers outperform the popular TV regularizer for image reconstruction.
As bilevel optimization is computationally quite intensive, it was proposed in \cite{KKHP2017} to unroll the forward-backward splitting (FBS) algorithm applied to \eqref{eq:VarProb} with a regularizer of the form \eqref{eq:convexridgereg}.
Accordingly, $R$ is optimized so that a predefined number $t$ of iterations of the FBS algorithm yields a good reconstruction.
Unfortunately, on a denoising task with learnable profiles $\psi_i$, the proposed approach does not match the performance of the bilevel optimization.

To deal with these shortcomings, we introduce an efficient framework\footnote{All experiments can be reproduced with the code published at \url{https://github.com/axgoujon/convex_ridge_regularizers}} to learn some $R$ of the form \eqref{eq:convexridgereg} with free-form convex profiles.
We train this $R$ on a generic denoising task and then plug it into~\eqref{eq:VarProb}.
This yields a generic reconstruction framework that is applicable to a variety of inverse problems.
The main contributions of the present work are as follows.
\begin{itemize}
    \item {\bf Interpretable and Expressive Model:} We use a one-hidden-layer neural network (NN) with learnable increasing linear-spline activation functions to parameterize $\V \nabla R$.
    We prove that this yields the maximal expressivity in the generic setting \eqref{eq:convexridgereg}.
    \item {\bf Embedding of the Constraints into the Forward Pass:} The structural constraints on $\V \nabla R$ are embedded into the forward pass during the training.
    This includes an efficient procedure to enforce the convexity of the profiles, and the computation of a bound on the Lipschitz constant of $\boldsymbol{\nabla} R$, which is required for our training procedure.
    \item {\bf Ultra-Fast Training:} The regularizer $R$ is learnt via the training of a multi-gradient-step denoiser.
    Empirically, we observe that a few gradient steps suffice to learn a best-performing $R$.
    This leads to training within a few minutes.
    \item {\bf Best Reconstruction Quality in a Constrained Scenario:} We show that our framework outperforms recent deep-learning-based approaches with comparable guarantees and constraints in two popular medical-imaging modalities (CT and MRI).
    This includes the PnP method with averaged denoisers and a variational framework with a learnable deep convex regularizer.
    This even holds for a strong mismatch in the noise level used for the training and the one found in the inverse problem.
\end{itemize}

\section{Architecture of the Regularizer}
In this section, we introduce the notions required to define the convex-ridge regularizer neural network (CRR-NN).
\subsection{General Setting}
Our goal is to learn a regularizer $R$ for the variational problem \eqref{eq:VarProb} that performs well across a variety of ill-posed problems.
Similar to the PnP framework, we view the denoising task
\begin{equation}
    \label{eq:denoise}
    \vec{x}^* = \argmin\limits_{\vec{x}\in\R^d} \frac{1}{2}\|\vec{x} - \vec{y}\|_2^2 + \lambda R(\vec{x})
\end{equation}
as the underlying base problem for training, where $\vec y$ is the noisy image.
Since we prioritize interpretability and reliability, we choose the simple convex-ridge regularizer~\eqref{eq:convexridgereg} and use its convolutional form.
More precisely, the regularity of an image $x$ is measured as
\begin{equation}
\label{eq:convconvexridgereg}
    R\colon x \mapsto \sum_{i=1}^{N_C}\sum_{\vec{k}\in\mathbb{Z}^2}\psi_i\bigl( (h_i * x)[\vec k]\bigr),
\end{equation}
where $h_i$ is the impulse response of a $2$D convolutional filter, $(h_i * x)[\vec k]$ is the value of the $\vec k$-th pixel of the filtered image $h_i * x$, and $N_C$ is the number of channels.
In the sequel, we mainly view the (finite-size) image $x$ as the (finite-dimensional) vector $\vec x \in\R^d$, and since \eqref{eq:convconvexridgereg} is a special case of \eqref{eq:convexridgereg}, we henceforth use the generic form \eqref{eq:convexridgereg} to simplify the notations. 
We use the notation $R_{\vec \param}$ to express the dependence of $R$ on the aggregated set of learnable parameters $\param$, which will be specified when necessary.
From now on, we assume that the convex profiles $\psi_i$ have Lipschitz continuous derivatives, i.e.\ $\psi_i\in C^{1,1}(\R)$.

\subsection{Gradient-Step Neural Network}
Given the assumptions on $R_{\param}$, the denoised image in~\eqref{eq:denoise} can be interpreted as the unique fixed point of $\boldsymbol{T}_{R_{\param},\lambda,\alpha}\colon \R^d \to \R^d$ defined by
\begin{align}
\label{eq:fixedopintoperator2}
    \boldsymbol{T}_{R_{\param},\lambda,\alpha}(\vec x) &= \vec x -\alpha\bigl((\vec x - \vec y) + \lambda \boldsymbol{\nabla} R_{\param}(\vec x)\bigr).
\end{align}
Iterations of the operator~\eqref{eq:fixedopintoperator2} implement a gradient descent with stepsize $\alpha$, which converges if $\alpha \in (0,2/(1+ \lambda L_{\param}))$, where $L_{\param} = \Lip(\boldsymbol{\nabla} R_{\param})$ is the Lipschitz constant of $\boldsymbol{\nabla} R_{\param}$.
In the sequel, we always enforce this constraint on $\alpha$.
The gradient of the generic convex-ridge expression~\eqref{eq:convexridgereg} is given by
\begin{equation}
\label{eq:gradmodel}
    \boldsymbol{\nabla} R_{\param}(\vec x) = \mat W^T \boldsymbol{\sigma}(\mat W \vec x),
\end{equation}
where $\mat W=[\vec w_1 \cdots \vec w_p]^T \in\R^{p\times d}$ and $\boldsymbol{\sigma}$ is a pointwise activation function whose components $(\sigma_i=\psi_i')_{i=1}^p$ are Lipschitz continuous and increasing.
In our implementation, the activation functions $\sigma_i$ are shared within each channel of $\mat W$.
The resulting gradient-step operator
\begin{align}
\label{eq:fixedopintoperator}
    \boldsymbol{T}_{R_{\param},\lambda,\alpha}(\vec x) &= (1 - \alpha) \vec x + \alpha \bigl(\vec y - \lambda \mat W^T \boldsymbol{\sigma}(\mat W\vec x)\bigr)
\end{align}
corresponds to a one-hidden-layer convolutional NN with a bias and a skip connection. We refer to it as a {\it gradient-step NN}. The training of a gradient-step NN will give a CRR-NN.

\section{Characterization of Good Profile Functions}
In this section, we provide theoretical results to motivate our choice of the profiles $\psi_i$ or, equivalently, of their derivatives $\sigma_i=\psi_i'$.
This will lead us to the implementation presented in Section \ref{sc:implementation}.

\subsection{Existence of Minimizers and Stability of the Reconstruction}
The convexity of $R_{\param}$ is not sufficient to ensure that the solution set in \eqref{eq:VarProb} is nonempty for a noninvertible forward matrix $\vec H$.
With convex-ridge regularizers, this shortcoming can be addressed under a mild condition on the functions $\psi_i$ (Proposition~\ref{pr:existenceCS}).
The implications for our implementation are detailed in Section~\ref{subsec:imporvedlearning}.
\begin{proposition}
\label{pr:existenceCS}
Let $\mat H\in\R^{m\times d}$ and $\psi_i\colon \R \rightarrow \R$, $i=1,\ldots,p$, be convex functions.
If $\argmin_{t\in\R}\psi_i(t)\neq \emptyset$ for all $i=1,\ldots,p$, then
\begin{equation}
\label{eq:invprboptim}
    \emptyset \neq \argmin \limits_{\vec x\in\R^d}\frac{1}{2}\|\mat H \vec x - \vec y\|_2^2 + \sum_{i=1}^p \psi_i(\vec w_i^T \vec x).
\end{equation}
\end{proposition}
\begin{proof}
Set $S_i = \argmin_{t\in\R}\psi_i(t)$.
Then, each ridge $\psi_i(\vec w_i^T \cdot)$ partitions $\R^d$ into the three (possibly empty) convex polytopes
\begin{itemize}
    \item $\Omega^i_{0} = \{\vec x \in \R^d: \vec w_i^T \vec x \in S_{i}\}$;
    \item $\Omega^i_{1} = \{\vec x \in \R^d: \vec w_i^T \vec x \leq \inf S_i\}$;
    \item $\Omega^i_{2} = \{\vec x \in \R^d: \vec w_i^T \vec x \geq \sup S_i\}$.
\end{itemize}
Based on these, we partition $\R^d$ into finitely many polytopes of the form $\bigcap_{i=1}^p \Omega^i_{m_i}$, where $m_i\in \{0,1,2\}$.
The infimum of the objective in \eqref{eq:invprboptim} must be attained in at least one of these polytopes, say, $P=\bigcap_{i=1}^p \Omega^i_{m_i}$.

Now, we pick a minimizing sequence $(\vec x_k)_{k\in\N} \subset P$.
Let $\mat M$ be the matrix whose rows are the rows of $\mat H$ and the $\vec w_i^T$ with $m_i \neq 0$.
Due to the coercivity of $\|\cdot\|_2^2$, we get that $\mat H \vec x_k$ remains bounded.
As the $\psi_i$ are convex, they are coercive on the intervals $(-\infty, \inf S_i]$ and $[\sup S_i, +\infty)$ and, hence, $\vec w_i^T\vec x_k$ also remains bounded.
Therefore, the sequence $(\mat M \vec x_k)_{k\in\N}$ is bounded and we can drop to a convergent subsequence with limit $\vec u \in \text{ran}(\vec M)$.
The associated set
\begin{equation}
Q=\{\vec x\in\R^d \colon \mat M\vec x = \vec u\} = \{\vec M^\dagger \vec u\} + \ker(\vec M)
\end{equation}
is a closed polytope.
It holds that
\begin{align}
    \mathrm{dist}(\vec x_k, Q) &= \mathrm{dist}\bigl(\vec M^\dagger \vec M \vec x_k + \mathrm P_{\ker(\vec M)}(\vec x_k), Q\bigr) \nonumber\\
    & \leq \mathrm{dist}(\vec M^\dagger \vec M \vec x_k, \vec M^\dagger \vec u) \to 0
\end{align}
as $k\to +\infty$ and, thus, that $\mathrm{dist}(P, Q)=0$.
The distance of the closed polytopes $P$ and $Q$ is 0 if and only if $P\cap Q \neq \emptyset$ \cite[Theorem~1]{W1968}.
Note that  $\psi_i(\vec w_i^T \cdot)$ is constant on $P$ if $m_i = 0$.
Hence, any $\vec x\in P\cap Q$ is a minimizer of \eqref{eq:invprboptim}.
\end{proof}

The proof of Proposition~\ref{pr:existenceCS} directly exploits the properties of ridge functions.
Whether it is possible to extend the result to more complex or even generic convex regularizers is not known to the authors.
The assumption in Proposition~\ref{pr:existenceCS} is rather weak as neither the cost function nor the one-dimensional profiles $\psi_i$ need to be coercive.
The existence of a solution for Problem \eqref{eq:VarProb} is a key step towards the stability of the reconstruction map in the measurement domain, which is given in Proposition \ref{pr:stability}.
\begin{proposition}
\label{pr:stability}
    Let $\mat H\in\R^{m\times d}$ and $\psi_i\colon \R \rightarrow \R$, $i=1,\ldots,p$, be convex, continuously differentiable functions with $\argmin_{t\in\R}\psi_i(t)\neq \emptyset$.
    For any $\vec y_1, \vec y_2 \in\mathbb{R}^m$ let
\begin{equation}
    \vec x_q \in \argmin \limits_{\vec x\in\R^d}\frac{1}{2}\|\mat H \vec x - \vec y_q\|_2^2 + \sum_{i=1}^p \psi_i(\vec w_i^T \vec x)
\end{equation}
with $q=1,2$ be the corresponding reconstructions.
Then,
\begin{equation}
    \|\mat H \vec x_1 - \mat H \vec x_2\|_2 \leq \| \vec y_1 - \vec y_2\|_2.
\end{equation}
\end{proposition}
\begin{proof}
    Proposition \ref{pr:existenceCS} guarantees the existence of $\vec x_q$.
    Since the objective in \eqref{eq:invprboptim} is smooth, it holds that $\mat H^T(\mat H \vec x_q -\vec y_q) + \V \nabla R(\vec x_q)=\vec 0$.
    From this, we infer that
    \begin{equation}
        \mat H^T\mat H (\vec x_1 - \vec x_2) + (\V \nabla R(\vec x_1) - \V\nabla R(\vec x_2))=\mat H^T(\vec y_1 - \vec y_2).
    \end{equation}
    Taking the inner product with $(\vec x_1 - \vec x_2)$ on both sides gives
    \begin{align}
        &\|\mat H \vec x_1 - \mat H \vec x_2\|_2^2 +(\vec x_1 - \vec x_2)^T(\V \nabla R(\vec x_1) - \V\nabla R(\vec x_2))\notag\\
        =&  (\mat H (\vec x_1 - \vec x_2))^T\mat (\vec y_1 - \vec y_2).
    \end{align}
    To conclude, we use the fact that the gradient of a convex map is monotone, i.e.\ $(\vec x_1 - \vec x_2)^T(\V \nabla R(\vec x_1) - \V\nabla R(\vec x_2))\geq 0$, and apply the Cauchy-Schwarz inequality to estimate
    \begin{equation}
        (\mat H (\vec x_1 - \vec x_2))^T\mat (\vec y_1 - \vec y_2)\leq \|\mat H \vec x_1 - \mat H \vec x_2\|\|\vec y_1 - \vec y_2\|.\qedhere
    \end{equation}
\end{proof}

\subsection{Expressivity of Profile Functions}
The gradient-step NN $\boldsymbol{T}_{R_{\param},\lambda,\alpha}$ introduced in ~\eqref{eq:fixedopintoperator} is the key component of our training procedure.
Here, we investigate its expressivity depending on the choice of the activation functions $\sigma_i$ used to parametrize $\boldsymbol{\nabla} R_{\param}$.

Let $C^{0,1}_{\uparrow}(\mathbb{R})$ be the set of scalar Lipschitz-continuous and increasing functions on $\mathbb{R}$, and let $\mathcal{LS}^m_{\uparrow}(\mathbb{R})$ be the subset of increasing linear splines with at most $m$ knots.
We also define
\begin{equation}
    \mathcal{E}(\mathbb{R}^d) = \bigl\{\mat W^T \boldsymbol{\sigma} (\mat W\cdot): \mat W\in\mathbb{R}^{p\times d}, \sigma_i \in C^{0,1}_{\uparrow}(\mathbb{R})\bigr\}
\end{equation}
and, further, for any $\Omega\subset\mathbb{R}^d$,
\begin{equation}
    \mathcal{E}(\Omega) = \bigl\{\left.\boldsymbol{f}\right|_{\Omega}\colon \boldsymbol{f}\in \mathcal{E}(\mathbb{R}^d)\bigr\}.
\end{equation}
In the following, we set $\Vert \boldsymbol{f}\Vert_{C(\Omega)} \coloneqq \sup_{\vec x \in \Omega} \Vert \boldsymbol{f}(\vec x) \Vert$ and $\Vert \boldsymbol{f}\Vert_{C^1(\Omega)} \coloneqq \sup_{\vec x \in \Omega} \Vert \boldsymbol{f}(\vec x) \Vert + \sup_{\vec x \in \Omega} \Vert \boldsymbol{J}_{\boldsymbol{f}}(\vec x) \Vert$.

The popular ReLU activation function is Lipschitz-continuous and increasing. Unfortunately, it comes with limited expressivity, as shown in Proposition~\ref{pr:expressivityactivationrelu}.
\begin{proposition}
\label{pr:expressivityactivationrelu}
Let $\Omega\subset \R^d$ be compact with a nonempty interior.
Then, the set
\begin{equation}\label{eq:setReLU}
   \bigl\{\mat W^T \mathrm{ReLU} (\mat W\cdot - \vec b)\colon \mat W\in\mathbb{R}^{p\times d}, \vec b\in\mathbb{R}^p\bigr\}
\end{equation}
is not dense with respect to $\Vert \cdot \Vert_{C(\Omega)}$ in $\mathcal{E}(\Omega)$.
\end{proposition}
\begin{proof}
Since $\Omega$ has a nonempty interior, there exists $\vec v \in \R^{d}$ with $\| \vec v \|_2=1$, $a \in \R$, and $\delta>0$ such that for $\boldsymbol{l}_{\vec v}\colon \R \to \R^d$ with $\boldsymbol{l}_{\vec v}(t)=t \vec v$, it holds that $\boldsymbol{l}_{\vec v}((a-\delta,a+\delta)) \subset \Omega$.
Now, we prove the statement by contradiction.
If the set \eqref{eq:setReLU} is dense in $\mathcal{E}(\Omega)$, then the set
\begin{align}
    &\bigl\{(\mat W \vec v)^T \mathrm{ReLU} (\mat W \vec v\cdot - \vec b)\colon \mat W\in\mathbb{R}^{p\times d}, \vec b\in\mathbb{R}^p\bigr\}\notag\\
    =&\biggl\{\sum_{i=1}^p w_i\mathrm{ReLU}(w_i \cdot - b_i)\colon w_i, b_i \in\mathbb{R}\biggr\}
\end{align}
is dense in $\mathcal{E}((a-\delta,a+\delta))$.
Note that all functions $f$ in \eqref{eq:setReLU} can be rewritten in the form
\begin{equation}\label{eq:special_form}
    f(x) =\! \sum_{i=1}^{p_1} \mathrm{ReLU}(w_i x - b_i) +\! \sum_{i=1}^{p_2} (-\mathrm{ReLU}(- \tilde w_i x - \tilde b_i)),
\end{equation}
where $w_i, \tilde w_i \in\mathbb{R}^+$, $ b_i, \tilde b_i \in\mathbb{R}$, and $p_1+p_2=p$.
Every summand in this decomposition is an increasing function.
For the continuous and increasing function
\begin{equation}
    g\colon t\mapsto  \mathrm{ReLU}(t - a + \delta/2) - \mathrm{ReLU}(t - a - \delta/2),
\end{equation}
the density implies that there exists $f$ of the form \eqref{eq:special_form} satisfying $\Vert g -f \Vert_{C((a-\delta,a+\delta))} \leq \delta/16$.
The fact that $g(a+\delta/2) = g(a+\delta)$ implies that $ (f(a+\delta) - f(a+\delta/2)) \leq \delta/8$.
In addition, it holds that
\begin{align}
& f(a+\delta) - f(a+\delta/2)\notag\\
\geq&  \sum_{i=1}^{p_1} \mathrm{ReLU}\bigl(w_i (a+\delta) - b_i\bigr) - \mathrm{ReLU}\bigl(w_i (a+\delta/2) - b_i\bigr)\notag\\
\geq& \sum_{\{i: b_i \leq w_i (a + \delta/2)\}} w_i (a+\delta - a - \delta/2)\notag\\
=& \sum_{\{i: b_i \leq w_i (a + \delta/2)\}} w_i \delta/2.
\end{align}
Hence, we conclude that $\sum_{\{i: b_i \leq w_i (a + \delta/2)\}} w_i \leq 1/4$.
Similarly, we can show that $\sum_{\{i: \tilde b_i \geq \tilde w_i (\delta/2 - a)\}} \tilde w_i \leq 1/4$.
Using these two estimates, we get that
\begin{align}
    \frac{7}{8} \delta &= g(a + \delta/2) - g(a - \delta/2) - \frac{1}{8}\delta \notag\\
    &\leq f(a + \delta/2) - f(a - \delta/2)\notag\\
    &\leq \sum_{\{i: b_i \leq w_i (a + \delta/2)\}} \delta w_i +  \sum_{\{i: \tilde b_i \geq \tilde w_i (\delta/2 - a)\}} \delta \tilde w_i \leq \frac{\delta}{2},
\end{align}
which yields a contradiction.
Hence, the set \eqref{eq:setReLU} cannot be dense in $\mathcal{E}(\Omega)$.
\end{proof}
\begin{remark}
    Any increasing linear spline $s$ with one knot is fully defined by the knot position $t_0$ and the slope on its two linear regions ($s_-$ and $s_+$).
    This can be expressed as $s=\vec u^T \mathrm{ReLU}(\vec u(t - t_0))$ with $\vec u = (\sqrt{s_+}, -\sqrt{s_-})$.
    Hence, among one-knot spline activation functions, the ReLU already achieves the maximal representational power for CRR-NNs.
    We infer that increasing PReLU and Leaky-ReLU induce the same limitations as the ReLU when plugged into CRR-NNs.
\end{remark}
In contrast, with Proposition \ref{pr:expressivityactivationls}, the set $\mathcal{E}(\Omega)$ can be approximated using increasing linear-spline activation functions.
\begin{proposition}
\label{pr:expressivityactivationls}
Let $\Omega\subset \R^d$ be compact and $m\geq 2$.
Then, the set
\begin{equation}\label{eq:JacLS}
   \bigl\{\mat W^T \boldsymbol{\sigma} (\mat W\cdot)\colon \mat W\in\mathbb{R}^{p\times d}, \sigma_i \in \mathcal{LS}^m_{\uparrow}(\mathbb{R})\bigr\}
\end{equation}
is dense with respect to $\Vert \cdot \Vert_{C(\Omega)}$ in $\mathcal{E}(\Omega)$.
\end{proposition}
\begin{proof}
First, we consider the case $d=1$.
By rescaling and shifting, we can assume that $S \subset [0,1]$ without loss of generality.
Let $f\in C^{0,1}_{\uparrow}([0,1])$, and $\varphi_{n}$ be the linear-spline interpolator of $f$ at locations $0,1/2^n,\ldots,(1-1/2^n),1$.
Since $f$ is increasing and $\varphi_{n}$ is piecewise linear, $\varphi_{n}$ is also increasing. Further, we get that
\begin{equation}
    \|f - \varphi_n\Vert_{C([0,1])} \leq \max_{k\in\{1,\ldots,2^n\}} f(k/2^n) - f((k-1)/2^n).
\end{equation}
Continuous functions on compact sets are uniformly continuous, which directly implies that $\Vert f - \varphi_n \Vert_{C([0,1])} \to 0$.
Now, we represent $\varphi_n$ as a linear combination of increasing linear splines with 2 knots
\begin{equation}
\label{eq:2knotwsexpansion}
    \varphi_n(x) = f(0) + \sum_{k=1}^{2^n} a_{k,n} g\bigl(2^n \cdot - (k-1)\bigr),
\end{equation}
where $a_{k,n}=(f(k/2^n) - f((k-1)/2^n))$ and $g$ is given by
\begin{equation}
    g(x)=\begin{cases}
    0, & x\leq 0 \\
    x, & 0<x\leq 1\\
    1,& \text{otherwise}.
    \end{cases}
\end{equation}
Finally, \eqref{eq:2knotwsexpansion} can be recast as $\varphi_n(x)=\vec w_n^T\boldsymbol{\sigma}_n(x \vec w_n)$, where each $\sigma_{n,i}$ is an increasing linear spline with 2 knots and $\vec w\in \mathbb{R}^{2^n}$.
This concludes the proof for $d=1$.

Now, we extend this result to any $d\in\mathbb{N}^{+}$.
Let $\boldsymbol{\Phi}\colon\mathbb{R}^d\rightarrow\mathbb{R}^d$ be given by $\vec x\mapsto \mat W^T\sigma(\mat W \vec x)$ with components \smash{$\sigma_i \in \mathcal{C}^0_{\uparrow}(\mathbb{R})$}.
Let $S_i=\{\vec w_i^T \vec x\colon \vec x\in\Omega\}$, where $\vec w_i \in \R^d$ is the $i$th row of $\mat W$.
Using the result for $d=1$, each $\sigma_i$ can be approximated in $C(S_i)$ by a sequence of functions $(\vec u_{n,i}^T \boldsymbol{\varphi}_{n}(\vec u_{n,i}\cdot))_{n\in\mathbb{N}}$, where $\boldsymbol{\varphi}_n$ has components $\varphi_{n,i} \in \mathcal{LS}^2_{\uparrow}(\mathbb{R})$ and $\vec u_{n,i}$ are vectors with a size that does not dependend on $i$.
Further, the $\vec u_{n,i}$ can be chosen such that the $j$th component is only nonzero for a single $i$.
Let $\mat U_n$ be the matrix whose columns are $\vec u_{n,i}$.
Then, we directly have that
\begin{equation}
    \lim_{n\to\infty}\max_{\vec x \in \{\vec y\in\R^d: y_i \in S_i\}} \bigl\|\mat U_n^T\boldsymbol{\varphi}_n(\mat U_n \vec x) - \boldsymbol{\sigma}(\vec x)\bigr\|_{2} = 0.
\end{equation}
Hence, the sequence of functions $((\mat U_n \mat W)^T\boldsymbol{\varphi}_n(\mat U_n \mat W\cdot))_{n\in\mathbb{N}}$ converges to $\boldsymbol{\Phi}$ in $C(\Omega)$.
This concludes the proof.
\end{proof}

In the end, Propositions~\ref{pr:expressivityactivationrelu} and~\ref{pr:expressivityactivationls} imply that using linear-spline activation functions instead of the ReLU for the $\sigma_i$ enables us to approximate more convex regularizers $R_{\param}$.
\begin{corollary}
Let $\Omega\subset \R^d$ be convex and compact with a nonempty interior.
Then, the regularizers of the form \eqref{eq:convexridgereg} with Jacobians of the form \eqref{eq:JacLS} are dense in
\begin{equation}\label{set:Reg}
    \biggl\{
    \sum_{i=1}^p \psi_i(\vec w_i^T \vec x): \psi_i \in C^{1,1}(\R) \text{ convex}, \vec w_i \in \R^d \biggr\}
\end{equation}
with respect to $\Vert \cdot \Vert_{C^1(\Omega)}$.
The density does not hold if we only consider regularizers with Jacobians of the form \eqref{eq:setReLU}.
\end{corollary}
\begin{proof}
Let $R$ be in \eqref{set:Reg}.
Consequently, its Jacobian is in $\mathcal{E}(\Omega)$.
Due to Proposition~\ref{pr:expressivityactivationrelu}, the regularizers with Jacobians of the form \eqref{eq:setReLU} cannot be dense with respect to 
$\Vert \cdot \Vert_{C^1(\Omega)}$.
Meanwhile, by Proposition~\ref{pr:expressivityactivationls}, we can choose $\vec x_0 \in \Omega$ and corresponding regularizers $R_n$ of the form \eqref{eq:convexridgereg} with $\boldsymbol{J}_{R_n} \in \eqref{eq:JacLS}$, $\Vert \boldsymbol{J}_{R_n} - \boldsymbol{J}_R \Vert_{C(\Omega)}\to 0$ as $n \to \infty$, and $R_n(\vec x_0) = R(\vec x_0)$.
Now, the mean-value theorem readily implies that $\Vert R_n - R \Vert_{C^1(\Omega)} \to 0$ as $n \to \infty$.
\end{proof}

Motivated by these  results, we propose to parameterize the $\sigma_i$ with learnable linear-spline activation functions.
This results in profiles $\psi_i$ that are splines of degree $2$, being piecewise polynomials of degree 2 with continuous derivatives.
\section{Implementation}
\label{sc:implementation}
\subsection{Training a Multi-Gradient-Step Denoiser}
\label{subsec:trainingprocedure}
Let $\{\vec x^m\}_{m=1}^M$ be a set of clean images and let $\{\vec y^m\}_{m=1}^M=\{\vec x^m + \vec n^m\}_{m=1}^M$ be their noisy versions, where $\vec n^m$ is the noise realisation.
Given a loss function $\mathcal L$, the natural procedure to learn the parameters of $R_{\param}$ based on \eqref{eq:denoise} is to solve
\begin{equation}\label{eq:TrainProb}
    \param^*_t, \lambda^*_t \in \argmin\limits_{\param,\lambda}\sum_{m=1}^M \mathcal L\bigl(\boldsymbol{T}^t_{R_{\param},\lambda,\alpha}(\vec y^m), \vec x^m\bigr)
\end{equation}
for the limiting case $t = \infty$ and an admissible stepsize $\alpha$.
Here, $\boldsymbol{T}^t_{R_{\param},\lambda,\alpha}$ denotes the $t$-fold composition of the gradient-step NN given in~\eqref{eq:fixedopintoperator}.
In principle, one can optimize the training problem \eqref{eq:TrainProb} with $t=\infty$. This forms a bilevel optimization problem that can be handled with implicit differentiation techniques \cite{CheRan2014,BKK2019,GilOngWil2021,ABM2023}.
However, it turns out that it is unnecessary to fully compute the
fixed-point $\boldsymbol{T}^\infty_{R_{\param},\lambda,\alpha}(\vec y^m)$ to learn $R_{\param}$ in our constrained setting.
Instead, we approximate $\boldsymbol{T}^\infty_{R_{\param},\lambda,\alpha}(\vec y^m)$ in a finite number of steps.
This specifies the \textit{$t$-step denoiser} NN $\boldsymbol{T}^t_{R_{\param},\lambda,\alpha}$, which is trained such that
\begin{equation}
\label{eq:tstepdf}
\boldsymbol{T}^t_{R_{\param},\lambda,\alpha}(\vec y^m) \simeq \vec x^m
\end{equation}
for $m=1,\ldots,M$.
This corresponds to a partial minimization of~\eqref{eq:denoise} with initial guess $\vec y^m$ or, equivalently, as the unfolding of the gradient-descent algorithm for $t$ iterations with shared parameters across iterations \cite{PraAgJac2020, AgHeMaJaco2019}. For small $t$, this yields a fast-to-evaluate denoiser. 
Since it is not necessarily a proximal operator, its interpretability is, however, limited.

Once the gradient-step NN is trained, we can plug the corresponding $R_{\param}$ into \eqref{eq:denoise}, and fully solve the optimization problem. This yields an interpretable \textit{proximal denoiser}.
In practice, turning a $t$-step denoiser into a proximal one requires the adjustment of $\lambda$ and the addition of a scaling parameter, as described in Section~\ref{subsec:scalingfactor}.
Our numerical experiments in Section~\ref{subsec:training} indicate that the number of steps $t$ used for training the multi-gradient-step denoiser has little influence on the test performances of both the $t$-step and proximal denoisers.
Hence, training the model within a few minutes is possible.
Note that our method bears some resemblance with the variational networks (VN) proposed in \cite{KKHP2017}, but there are some fundamental differences.
While the model used in \cite{KKHP2017} also involves a sum of convex ridges with learnable profiles, these are parameterized by radial-basis functions and only the last step of the gradient descent is included in the forward pass.
The authors of \cite{KKHP2017} observed that an increase in $t$ deters the denoising performances, which is not the case for our architecture.
More differences are outlined in Section \ref{subsec:imporvedlearning}.

\subsection{Implementation of the Constraints}
\label{subsec:imporvedlearning}
Our learning of the $t$-step denoiser is constrained as follows.
\begin{enumerate}[label=(\roman*)]
    \item The activation functions $\sigma_i$ must be increasing (convexity constraint on $\psi_i$).
    \item The activation functions $\sigma_i$ must take the value 0 somewhere (existence constraint).
    \item The stepsize in \eqref{eq:fixedopintoperator} should satisfy $\alpha \in (0,2/(1+ \lambda L_{\param}))$ (convergent gradient-descent).
\end{enumerate}
Since the methods to enforce these constraints can have a major impact on the final performance, they must be designed carefully.
\paragraph{Monotonic Splines}
Here, we address Constraints (i) and (ii) simultaneously.
Similar to \cite{BCGA2020,bohra2021learning}, we use learnable linear splines $\sigma_{\vec c^i}\colon\mathbb{R}\rightarrow\mathbb{R}$ with $(M+1)$ uniform knots $\nu_m = (m - M/2)\Delta$, $m=0,\ldots,M$, where $\Delta$ is the spacing of the knots.
For simplicity, we assume that $M$ is even.
The learnable parameter $\vec c^i=(c_m^i)_{m=0}^{M}\in\mathbb{R}^{M+1}$ defines the value $\sigma_{\vec c^i}(\nu_m)=c^i_m$ of $\sigma_{\vec c^i}$ at the knots.
To fully characterize $\sigma_{\vec c^i}$, we extend it by the constant value $\vec c^i_0$ on $(-\infty,\nu_0]$ and $\vec c^i_M$ on $[\nu_M, +\infty)$.
This choice results in a linear extension for the corresponding indefinite integrals that appear for the regularizer $R_{\param}$ in \eqref{eq:denoise}.
Further details on the implementation of learnable linear splines can be found in \cite{BCGA2020}.

Let $\mat D\in \mathbb{R}^{M\times(M+1)}$ be the one-dimensional finite-difference matrix with $(\mat D \vec c^i)_m = c_{m+1}^i - c_m^i$ for $m=0, \ldots, (M-1)$.
As $\sigma_{\vec c^i}$ is piecewise-linear, it holds that
\begin{equation}
    \text{$\sigma_{\vec c^i}$ is increasing} \Leftrightarrow \mat D \vec c^i\geq 0.
\end{equation}
In order to optimize over $\{\sigma_{\vec c}\colon \mat D \vec c\geq 0\}$, we reparameterize the linear splines as $\sigma_{\boldsymbol{P}_{\uparrow}(\vec c^i)}$, where
\begin{equation}
\label{eq:projection}
    \boldsymbol{P}_{\uparrow} = \mat C \mat D^{\dagger}\mathrm{ReLU}(\mat D\,\cdot)
\end{equation}
is a nonlinear projection operator onto the feasible set. There, $\mat D^{\dagger}$ denotes the Moore-Penrose inverse of $\mat D$ and $\mat C= (\mathbf{Id}_{M+1} - \vec 1_{M+1} \vec e_{M/2+1}^T)$ shifts the output such that the $(M/2+1)$th element is zero.
In effect, this projection simply preserves the nonnegative finite differences between entries in $\vec c^i$ and sets the negative ones to zero.
As the associated profiles $\psi_i$ are convex and satisfy $\psi'_i(0)=\sigma_i(0)=0$, Proposition~\ref{pr:existenceCS} guarantees the existence of a solution for Problem \eqref{eq:VarProb}.

The proposed parameterization $\sigma_{\boldsymbol{P}_{\uparrow}(\vec c^i)}$ of the splines has the advantage to use unconstrained trainable parameters $\vec c_i$.
The gradient of the objective in \eqref{eq:TrainProb} with respect to $\vec c_i$ directly takes into account the constraint via $\boldsymbol{P}_{\uparrow}$.
This approach differs significantly from the more standard projected gradient descent---as done in \cite{KKHP2017} to learn convex profiles---where the $\V c_i$ would be projected onto $\{\vec c_i \colon \mat D \vec c_i\geq 0\}$ after each gradient step.
While the latter routine is efficient for convex problems, we found it to perform poorly for the non-convex problem \eqref{eq:TrainProb}.
For an efficient forward and backward pass with auto-differentiation, $\boldsymbol{P}_{\uparrow}$ is implemented with the $\texttt{cumsum}$ function instead of an explicit construction of the matrix $\mat D^{\dagger}$, and the computational overhead is very small.

\paragraph{Sparsity-Promoting Regularization}
The use of learnable activation functions can lead to overfitting and can weaken the generalizability to arbitrary operators $\mat H$.
Hence, the training procedure ought to promote simple linear splines.
Here, it is natural to promote the better-performing splines with the fewest knots.
This is achieved by penalizing the second-order total variation $\|\mat L \boldsymbol{P}_{\uparrow}(\vec c_i)\|_1$ of each spline $\sigma_{\boldsymbol{P}_{\uparrow}(\vec c_i)}$, where $\mat L \in \R^{(M-1)\times(M+1)}$ is the second-order finite-difference matrix.
The final training loss then reads
\begin{equation}\label{eq:TrainProbReg}
\sum_{m=1}^M \mathcal L\bigl(\boldsymbol{T}^t_{R_{\param},\lambda,\alpha}(\vec y^m), \vec x^m\bigr) + \eta \sum_{i=1}^p \|\mat L \boldsymbol{P}_{\uparrow}(\vec c_i)\|_1,
\end{equation}
where $\eta\in\mathbb{R}^{+}$ allows one to tune the strength of the regularization.
We refer to \cite{Unser2019} for more theoretical insights into second-order total-variation regularization and to \cite{BCGA2020} for experimental evidence of its relevance for machine learning.

\paragraph{Convergent Gradient Steps}
\label{subsc:Lipsestimate}
Constraint (iii) guarantees that the $t$-fold composition of the gradient-step NN $\boldsymbol{T}^t_{R_{\param},\lambda,\alpha}$ computes the actual minimizer of~\eqref{eq:denoise} for $t\to\infty$. Therefore, it should be enforced in any sensible training method.
In addition, it brings stability to the training.
To fully exploit the model capacity, even for small $t$, we need a precise upper-bound for $\Lip(\boldsymbol{\nabla} R_{\param})$.
The estimate that we provide in Proposition~\ref{pr:LipsBound} is sharper than the classical bound derived from the sub-multiplicativity of the Lipschitz constant for compositional models. It is easily computable as well.

\begin{proposition}
\label{pr:LipsBound}
Let $L_{\theta}$ denote the Lipschitz constant of $\boldsymbol{\nabla} R_{\param}(\vec x) = \mat W^T \boldsymbol{\sigma}(\mat W \vec x)$ with $\mat W\in\R^{p\times d}$ and $\sigma_i \in \mathcal{C}^{0,1}_{\uparrow}(\mathbb{R})$.
With the notation $\mat \Sigma_{\infty}=\mathbf{diag}(\|\sigma'_1\|_{\infty}, \ldots, \|\sigma'_p\|_{\infty})$ it holds that
\begin{equation}\label{eq:bound_imp}
    L_{\theta} \leq \|\mat W^T \mat \Sigma_{\infty} \mat W\| = \|\sqrt{\mat \Sigma_{\infty}}\mat W\|^2,
\end{equation}
which is tighter than the naive bound
\begin{equation}\label{eq:loosebound}
   L_{\theta} \leq L_{\boldsymbol{\sigma}}\|\mat W\|^2.
\end{equation}
\end{proposition}
\begin{proof}
The bound \eqref{eq:loosebound} is a standard result for compositional models. 
Next, we note that the Hessian of $R_{\param}$ reads
\begin{equation}
    \mat H_{R_{\param}}(\vec x) = \mat W^T \mat \Sigma(\mat W \vec x) \mat W,
\end{equation}
where $\mat \Sigma(\vec z) = \mathbf{diag}(\sigma'_1(z_1), \ldots, \sigma'_p(z_p))$.
Further, it holds that $L_{\theta} \leq \sup_{\vec x \in \R^d} \Vert \mat H_{R_{\param}}(\vec x) \Vert$.
Since the functions $\sigma_i$ are increasing, we have for every $\vec x\in\R^p$ that $\mat \Sigma_\infty - \mat \Sigma(\mat W\vec x)\succeq 0$
and, consequently,
\begin{equation}
    \mat W^T \bigl(\mat \Sigma_\infty - \mat \Sigma(\mat W\vec x)\bigr) \mat W \succeq 0.
\end{equation}
Using the Courant-Fischer theorem, we now infer that the largest eigenvalue of $\mat W^T \mat \Sigma_\infty \mat W$ is greater than that of $\mat W^T \mat \Sigma(\mat W \vec x) \mat W$.
\end{proof}

The bounds \eqref{eq:bound_imp} and \eqref{eq:loosebound}are in agreement when the activation functions are identical, which is typically not the case in our framework.
For the 14 NNs trained in Section~\ref{sc:experiments}, we found that the improved bound \eqref{eq:bound_imp} was on average 3.2 times smaller than \eqref{eq:loosebound}.
As \eqref{eq:bound_imp} depends on the parameters of the model, it is critical to embed the computation into the forward pass. Otherwise, the training gets unstable.
This is done by first estimating the normalized eigenvector $\mat u$ corresponding to the largest eigenvalue of $\mat W^T \mat \Sigma_{\infty} \mat W$ via the power-iteration method in a non-differentiable way, for instance under the \texttt{torch.no\_grad()} context-manager.
Then, we directly plug the estimate $L_{\theta}\simeq \|\mat W^T \mat \Sigma_{\infty} \mat W \mat u\|$ in our model and hence embed it in the forward pass.
This approach is inspired by the spectral-normalization technique proposed in \cite{MKKY2018}, which is a popular and efficient way to enforce Lipschitz constraints on fully connected linear layers.
Note that a similar simplification is also proposed and studied in the context of deep equilibrium models~\cite{fung2022jfb}. In practice, the estimate $\vec u$ is stored so that it can be used as a warm start for the next computation of $L_{\theta}$.
\subsection{From Gradients to Potentials}
To recover the regularizer $R$ from its gradient $\V \nabla R$, one has to determine the profiles $\psi_i$, which satisfy $\psi_i'=\sigma_{\boldsymbol{P}_{\uparrow}(\vec c^i)}$.
Hence, each $\psi_i$ is a piecewise polynomial of degree 2 with continuous derivatives, i.e.\ a spline of degree two.
These can be expressed as a weighted sum of shifts of the rescaled causal B-spline of degree $2$ \cite{unser1999splines}, more precisely as
\begin{equation}
\label{eq:profilesplines}
    \psi_i = \sum_{k\in\mathbb{Z}} d_k^i \beta_+^2\left(\frac{\cdot - k}{\Delta}\right).
\end{equation}
To determine the coefficients $(d_k^i)_{k\in\mathbb{Z}}$, we use the fact that $(\beta_+^2)'(k) = (\delta_{1,k} - \delta_{2,k})$, where $\delta$ is the Kronecker delta, see \cite{unser1999splines} for details.
Hence, we obtain that $d_k^i - d_{k-1}^i = (\boldsymbol{P}_{\uparrow}(\vec c^i))_k$, which defines $(d_k^i)_{k\in\mathbb{Z}}$ up to a constant.
This constant can be set arbitrarily as it does not affect $\V \nabla R$.
Due to the finite support of $\beta_+^2$, one can efficiently evaluate $\psi_i$ and then $R$.
\subsection{Boosting the Universality of the Regularizer}
\label{subsec:scalingfactor}
The learnt $R_{\param}$ depends on the training task (denoising) and on the noise level.
To solve a generic inverse problem, in addition to the regularization strength $\lambda$, we propose to incorporate a tunable scaling parameter $\mu\in\mathbb{R}^+$ and to compute
\begin{equation}
\label{eq:scalingfactor}
    \argmin \limits_{\vec x\in\mathbb{R}^d}\frac{1}{2}\|\mat H\vec x -\vec y\|_2^2 + \lambda/\mu R_{\param}(\mu \vec x).
\end{equation}
While the scaling parameter is irrelevant for homogeneous regularizers such as the Tikhonov and TV, it is known to boost the performance within the PnP framework when applied to the input of the denoiser \cite{xu2020boosting}.
During the training of $t$-step denoisers, we also learn a scaling parameter $\mu$ by letting the gradient step NN \eqref{eq:fixedopintoperator2} become
\begin{align}
    \boldsymbol{T}_{R_{\param},\lambda,\mu,\alpha}(\vec x) &= \vec x -\alpha\bigl((\vec x - \vec y) + \lambda \boldsymbol{\nabla} R_{\param}(\mu \vec x)\bigr),
\end{align}
with now $\alpha < 2/(1 + \lambda \mu \mathrm{Lip}(\V \nabla R_{\param}))$.
\subsection{Reconstruction Algorithm}
The objective in~\eqref{eq:scalingfactor} is smooth with Lipschitz-continuous gradients.
Hence, a reconstruction can be computed through gradient-based methods.
We found the fast iterative shrinkage-thresholding algorithm (FISTA, Algorithm \ref{alg:FISTA}) to be well-suited to the problem while it also allows us to enforce the positivity of the reconstruction. Other efficient algorithms for CRR-NNs include the adaptive gradient descent (AdGD)~\cite{malitsky2020a} and its proximal extension \cite{latafat2023adaptive}; both benefit from a stepsize based on an estimate of the local Lipschitz constant of $\V \nabla R$ instead of a more conservative global one.
\begin{algorithm}[t]
	\begin{algorithmic}
	    \State \textbf{Input:} $\vec x_0\in\R^d$, $\vec y\in\R^m$, $\lambda \geq 0$, $\mu > 0$
		\State Set $k=0$, $\vec z_0=\vec x_0$, $\alpha = 1/(\mu \lambda \mathrm{Lip}(\V\nabla R) + \|\mat H\|^2)$, $t_0=1$
   
		\While{tolerance not reached}
            \State $\vec x_{k+1} = (\vec z_{k} - \alpha (\mat H^T (\mat H \vec z_k - \vec y) + \lambda\V\nabla R(\mu \vec z_k)))_+$
            \State $t_{k+1} = (1 + \sqrt{4 t_k^2 + 1})/2$
		\State $\vec z_{k+1} = \vec x_{k+1} + \frac{t_k - 1}{t_{k+1}}(\vec x_{k+1} - \vec x_{k})$
            \State $k \gets k + 1$
        \EndWhile
		\State \textbf{Output:} $\vec x_{k}$
		\caption{FISTA \cite{beck2009fast} to solve \eqref{eq:scalingfactor}}
		\label{alg:FISTA}
	\end{algorithmic}
\end{algorithm}
\section{Connections to Deep-Learning Approaches}
\label{sec:deepframeworks}
Our proposed CRR-NNs have a single nonlinear layer, which is rather unusual in an the era of deep learning.
To further explore their theoretical properties, we briefly discuss two successful deep-learning methods, namely, the PnP and the explicit design of convex regularizers, and state their most stable and interpretable versions.
This  will clarify the notions of strict convergence, interpretability, and universality.
All the established comparisons are synthesized in Table~\ref{tab:PropReg}.
\begin{table}[t]
\centering
\caption{Properties of different regularization frameworks.\label{tab:PropReg}}
\setlength\tabcolsep{5.5pt}
\begin{tabular}{llllll}

\toprule
         & Explicit & Provably & Universal & Shallow & Smooth\\ 
        & cost & convergent &  &  & reg.\\ 
        
\midrule
TV       & \cmarkc   & \cmarkc    & \cmarkc    & \cmarkc & \xmarkc          \\ 
ACR      & \cmarkc   & \cmarkc    & \xmarkc   & \xmarkc   & \xmarkc         \\
DnICNN & \cmarkc   & \cmarkc    & \cmarkc    & \xmarkc  & \cmarkc        \\
PnP-{$\beta$}CNN\!\! & \xmarkc   & \cmarkc    & \cmarkc    & \xmarkc  & \color{red}{-}         \\
PnP-DnCNN\!\! & \xmarkc   & \xmarkc    & \cmarkc    & \xmarkc  & \color{red}{-}         \\
CRR-NN     & \cmarkc   & \cmarkc    & \cmarkc    & \cmarkc   & \cmarkc           \\

\bottomrule
\\
\end{tabular}
\end{table}
\subsection{Plug-and-Play and Averaged Denoisers}
\paragraph{Convergent Plug-and-Play}

The training procedure proposed for CRR-NNs leads to a convex regularizer $R_{\param}$, whose proximal operator \eqref{eq:denoise} is a good denoiser.
Conversely, the proximal operator can be replaced by a powerful denoiser $\V D$ in proximal algorithms, which is referred to as PnP.
In the PnP-FBS  algorithm derived from \eqref{eq:VarProb} \cite{ComWaj2005,beck2009fast}, the reconstruction is carried out iteratively via
\begin{equation}
    \label{eq:PnPFBS}
    \vec x_{k+1} = \V D\bigl(\vec x_{k} - \alpha \mat H^T (\mat H \vec x_k - \vec y)\bigr),
\end{equation}
where $\alpha$ is the stepsize and $\V D\colon \mathbb{R}^d\rightarrow\mathbb{R}^d$ is a generic denoiser.
A standard set of sufficient conditions\footnote{Here, $\mat H$ can be noninvertible; otherwise, weaker conditions exist \cite{ryu2019plug}.} to guarantee convergence of the iterations \eqref{eq:PnPFBS} is that
\begin{enumerate}[label=(\roman*)]
    \item \label{it:prop2}  $\V D$ is averaged, namely $\V D=\beta \V N + (1-\beta)\mathbf{Id}$ where $\beta\in (0,1)$ and $\V N\colon \mathbb{R}^n\rightarrow\mathbb{R}^n$ is a nonexpansive mapping;
    \item \label{it:prop1} $\alpha \in [0,2/\Vert \mat H \Vert^2)$;
    \item the update operator in \eqref{eq:PnPFBS} has a fixed point.
\end{enumerate}

In general, Condition (i) is not sufficient to ensure that $\V D$ is the proximal operator of some convex regularizer $R$.
Hence, its interpretability is still limited.
Further, Condition (ii) implies that $\vec x \mapsto \left( \vec x - \alpha \mat H^T (\mat H \vec x - \vec y) \right)$ is averaged.
Hence, as averagedness is preserved through composition, the iterates are updated by the application of an averaged operator (see \cite{HNS2021} for details).
With Condition (iii), the convergence of the iterations \eqref{eq:PnPFBS} follows from Opial's convergence theorem. Beyond convergence, it is known that averaged denoisers with $\beta\leq 1/2$ yield a stable reconstruction map in the measurement domain \cite{ducotterd2022improving}, in the same sense as given in Proposition \ref{pr:stability} for CRR-NNs.

The nonexpansiveness of $\V D$ is also commonly assumed for proving the convergence
of other PnP schemes.
This includes, for instance, gradient-based PnP \cite{ABM2023}.
There, the gradient $\boldsymbol{\nabla}R$ of the regularizer used in reconstruction algorithms is replaced with a learned monotone operator $\V F = \mat I -\V D$.
The operator $\V D$ can be interpreted as a denoiser and is assumed to be nonexpansive to prove convergence.
\paragraph{Constraint vs Performance}
As discussed in \cite{chan2016plug,HurLec2022}, the performance of the denoiser $\V D$ is in direct competition with its averagedness.
 A simple illustration of this issue is provided in Figure~\ref{fig:denoisevslip}.
Unsurprisingly, Condition~(i) is not met by any learnt state-of-the-art denoiser, and it is usually also relaxed in the PnP literature.

\begin{figure}[t]
    \centering
    \includegraphics[width=0.2\textwidth]{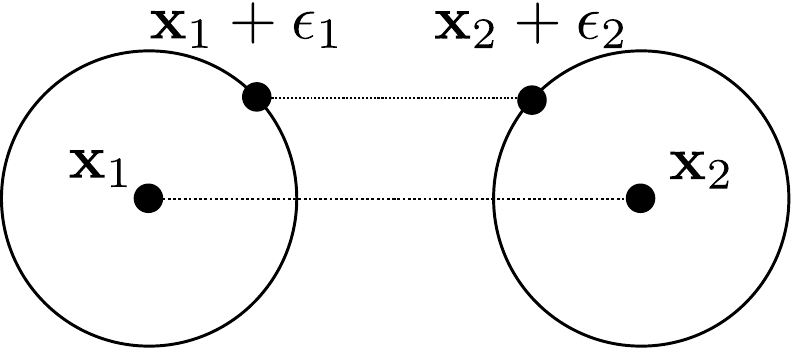}
    \caption{The distance between the two noisy images $(\vec x_1+\vec \epsilon_1)$ and $(\vec x_2 + \vec \epsilon_2)$ can be smaller than that between their clean versions $\vec x_1$ and $\vec x_2$. This limits the performance of a nonexpansive denoiser $\V D$ since $\|\V D(\vec x_1+\vec \epsilon_1) - \V D(\vec x_2+\vec \epsilon_2)\|\leq \|\vec x_1+\vec \epsilon_1 - (\vec x_2+\vec \epsilon_2)\|< \|\vec x_1 - \vec x_2\|$ in the scenario depicted.}
    \label{fig:denoisevslip}
\end{figure}
 For instance, it is common to use non-$1$-Lipschitz learning modules, such as batch normalization \cite{ryu2019plug}, or to only constrain the residual $(\mathbf{Id}-\V D)$ to be nonexpansive, which enables one to train a nonexpansive NN in a residual way \cite{ryu2019plug, LAWK2021, HNS2021}, with the caveat that $\mathrm{Lip}(\V D)$ can be as large as 2.
 Another recent approach consists of penalizing during training either the norm of the Jacobian of $\V D$ at a finite set of locations~\cite{PRTW2021,HurLec2022} or of another local estimate of the Lipschitz constant~\cite{AM2022, ABM2023}.
 Interestingly, even slightly relaxed frameworks usually yield significant improvements in the reconstruction quality.
 However, they do not provide convergence guarantees for ill-posed inverse problems, which is problematic for sensitive applications such as biomedical imaging.

\paragraph{Averaged Deep NNs}
To leverage the success of deep learning, $\V N$ is typically chosen as a deep CNN of the form\footnote{The benefit of standard skip connections combined with the preservation of the nonexpansiveness of the NN is unclear.}
\begin{equation}\label{eq:ArchitectureCNN}
     \V N = \V C_K\circ \V \sigma\circ\cdots\circ \V C_2\circ \V \sigma\circ \V C_1,
\end{equation}
where $\V C_k$ are learnable convolutional layers and $\V \sigma$ is the activation function \cite{MKKY2018,ryu2019plug,bohra2021learning}.
To meet Condition (i), $\V N$ must be nonexpansive, which one usually achieves by constraining $\V C_k$ and $\V \sigma$ to be nonexpansive.
This is predicated on the sub-multiplicativity of the Lipschitz constant with respect to composition; as in $\mathrm{Lip}(\V f\circ \V g)\leq \mathrm{Lip}(\V f)\mathrm{Lip}(\V g)$.
Unfortunately, this bound is not sharp and may grossly overestimate $\mathrm{Lip}(\V f\circ \V g)$.
For deep models, this overestimation aggravates since the bound is used sequentially.
Therefore, for averaged NNs, the benefit of depth is unclear because the gain of expressivity brought by the many layers is reduced by a potentially very pessimistic Lipschitz-constant estimate.
Put differently, these CNNs can easily learn the zero function while they struggle to generate mappings with a Lipschitz constant close to one.
For the same reason, the learning process is also prone to vanishing gradients in this constrained setting. Under Lipschitz constraints, the zero-gradient region of the popular ReLU activation function causes provable limitations \cite{HCC2018, ALG19, NGBU2022}.
Some of these can be resolved by the use of PReLU activation functions instead.

In this work, CRR-NNs are compared against two variants of PnP.
\begin{itemize}
    \item {\bf PnP-DnCNN} corresponds to the popular implementation given in \cite{ryu2019plug}. The denoiser is a DnCNN with 1-Lipschitz linear layers (the constraints are therefore enforced on the residual map only) and unconstrained batch-normalization modules. Hence this method has no convergence and stability guarantees, especially for ill-posed inverse problems.
    \item {\bf PnP-$\beta$CNN} corresponds to PnP equipped with a provably averaged denoiser. This method comes with similar guarantees as CRR-NNs but less interpretability. It is included to convey the message that the standard way of enforcing Lipschitz constraints affects expressivity as reported for instance in \cite{chand2023multiscale}, and even makes it hard to improve upon TV. With that in mind, CRR-NNs provide a way to overcome this limitation.
\end{itemize}
\paragraph{Construction of Averaged Denoisers from CRR-NNs} The training of CRR-NNs offers two ways to build averaged denoisers.
Since proximal operators are half-averaged, we directly get that the proximal denoiser~\eqref{eq:denoise} is an averaged operator.
For the $t$-step denoiser, the following holds.
\begin{proposition}
\label{pr:averagedoperator}
The $t$-step denoiser~\eqref{eq:tstepdf} is averaged for $\alpha\in[0,2/(2+\lambda L_{\param})]$ with $L_{\param} = \Lip(\boldsymbol{\nabla} R_{\param})$.
\end{proposition}
\begin{proof}
The $t$-step denoiser is built from the gradient-step operator $\boldsymbol{T}_{R_{\param},\lambda,\alpha}$.
Here, we use the more explicit notation
\begin{equation}
    \boldsymbol{T}(\vec x, \vec y) = \vec x - \alpha((\vec x - \vec y) + \lambda \boldsymbol{\nabla}R_{\theta}(\vec x)).
\end{equation}
This makes explicit the dependence on $\vec y$ and, for simplicity, the dependence on $R_{\param}$, $\lambda$, and $\alpha$ are omitted.
It is known that $\boldsymbol{T}$ is averaged with respect to $\vec x$ for $\alpha\in(0,2/(1+\lambda L_{\param}))$.
This ensures convergence of gradient descent, but it does not characterize the denoiser itself.
The $t$-step denoiser depends on the initial value $\vec x_0=\vec y$ and is determined by the recurrence relation $\vec x_{k+1} = \boldsymbol{T}(\vec x_k, \vec y)$.
For the map $\boldsymbol{L}_k\colon \vec y \mapsto \vec x_{k}$, it holds that
$\boldsymbol{L}_{k+1} = \boldsymbol{U}\circ \boldsymbol{L}_k + \alpha \mathbf{Id}$,
where $\boldsymbol{U}=\mathbf{Id} - \alpha (\mathbf{Id} + \lambda\boldsymbol{\nabla}R_{\theta})$.
The Jacobian of $\boldsymbol{U}$ reads $\mathbf{J}_{\boldsymbol{U}}=\mathbf{I} - \alpha (\mathbf{I} + \lambda\mathbf{H}_{R_{\theta}})$ and satisfies that $((1-\alpha) - \alpha\lambda L_{\param}) \mat I \preceq \mathbf{J}_{\boldsymbol{U}}\preceq (1-\alpha)\mat I$.
From this, we infer that
\begin{equation}
    \mathrm{Lip}(\mathbf{U})\leq \max \bigl(\alpha\lambda L_{\param} - (1 - \alpha), 1-\alpha\bigr).
\end{equation}
Since $\alpha\leq 2/(2 + \lambda L_{\param})$, we then get that $\mathrm{Lip}(\mathbf{U})\leq (1-\alpha)$.
Hence, $\mathrm{Lip}(\mathbf{U}\circ \boldsymbol{L}_k)\leq (1-\alpha)\mathrm{Lip}(\boldsymbol{L}_{k})$.
Since $\boldsymbol{L}_{0}=\mathbf{Id}$ is averaged, the same holds by induction for all the $t$-step denoisers $\boldsymbol{L}_{t}$.
\end{proof}

Note that for $\alpha\in (2/(2+\lambda L_{\param}),2/(1+\lambda L_{\param}))$, the $1$-step denoiser is also averaged but, for $1<t<+\infty$, it remains an open question.
The structure of $t$-step and proximal denoisers differs radically from averaged CNNs as in~\eqref{eq:ArchitectureCNN}.
For instance, the $t$-step denoiser uses the noisy input $\vec y$ in each layer.
Remarkably, these skip connections preserve the averagedness of the mapping.
While constrained deep CNNs struggle to learn mappings that are not too contractive, both proximal and $t$-step denoisers can easily reproduce the identity by choosing $R_{\param}=0$.
This seems key to account for the fact that the proposed denoisers outperform averaged deep NNs, while they can be trained two orders of magnitude faster, see Section~\ref{sc:experiments}.
\subsection{Deep Convex Regularizers}
Another approach to leverage deep-learning-based priors with stability and convergence guarantees consists of learning a deep convex regularizer $R$.
These priors are typically parameterized with an ICNN, which is a NN with increasing and convex activation functions along with positive weights for some linear layers \cite{AXK2016}.
There exist various strategies to train the ICNN.

The adversarial convex regularizer (ACR) framework \cite{MukDit2021,MukSch2021} relies on the adversarial training proposed in \cite{lunz2018adversarial}.
The regularizer is learnt by minimizing its value on clean images and maximizing its value on
unregularized reconstructions.
This allows for learning non-smooth $R$ and also avoids bilevel optimization.
A key difference with CRR-NNs and PnP methods is that ACR is modality-depend (it is not universal).
In addition, with $R$ being non-smooth, it is challenging to exactly minimize the cost function, but the authors of \cite{MukDit2021,MukSch2021} did not find any practical issues in that matter using gradient-based solvers.
To boost the performance of $R$, they also added a sparsifying filter bank to the ICNN, namely, a convex term of the form $\|\mat U \vec x\|_1$, where the linear operator $\mat U$ is made of convolutions learnt conjointly with the ICNN.

In \cite{cohen2021has}, the regularizer is trained so that its gradient step is a good blind Gaussian denoiser. There, the authors use ELU activations in the ICNN\footnote{The authors also explore non-convex regularization but they offer no guarantees on computing the global minimum.} to obtain a smooth $R$.

The aforementioned ICNN-based frameworks \cite{MukDit2021, MukSch2021, cohen2021has} have major differences with CRR-NNs: (i) they typically require orders of magnitude more parameters; (ii) the computation of $\V \nabla R$, used to solve inverse problems, requires one to back-propagate through the deep CNN which is time-consuming; (iii) the role of each parameter is not interpretable because of the depth of the model (see Section \ref{subsec:underthehood}). As we shall see, CRR-NNs are much faster to train and tend to perform better (see Section \ref{sc:experiments}).

\section{Experiments}
\label{sc:experiments}
\subsection{Training of CRR-NNs}
\label{subsec:training}
The CRR-NNs are trained on a Gaussian-denoising task with noise levels $\sigma\in\{5/255, 25/255\}$.
The same procedure as in \cite{ZZCMZ2017,ryu2019plug} is used  to form 238,400 patches of size $(40\times 40)$ from 400 images of the BSD500 dataset \cite{arbelaez_contour_2011}.
For validation, the same 12 images as in \cite{ZZCMZ2017,ryu2019plug} are used.
The weights $\mathbf{W}$ in $R_{\param}$ are parameterized as the composition of two zero-padded convolutions with kernels of size $(7\times 7)$ and with $8$ and $32$ output channels, respectively.
This composition of two linear components, although not more expressive theoretically, facilitates the patch-based training of CRR-NNs.
For inference, the convolutional layer can then be transformed back to a single convolution.
Similar to \cite{CheRan2014}, the kernels of the convolutions are constrained to have zero mean.
Lastly, the linear splines have $M+1=21$ equally distant knots with $\Delta = 0.01$, and the sparsifying regularization parameter is $\eta=\num{2e-3}(255\sigma )$. We initially set $\vec c_i =\vec 0$.

The CRR-NNs are trained for 10 epochs with $t\in\{1,2,5,10,20,30,50\}$ gradient steps.
For this purpose, the $\ell_1$ loss is used for $\mathcal L$ along with the Adam optimizer with its default parameters $(\beta_1, \beta_2)=(0.9,0.999)$, and the batch size is set to $128$.
The learning rates are decayed with rate $0.75$ at each epoch and initially set to $0.05$ for the parameters $\lambda$ and $\mu$, to $\num{e-3}$ for $\mat W$, and to $\num{5e-5}$ for $\vec c_i$.

Recall that for a given $t$, the training yields two denoisers.
\begin{itemize}
    \item {\bf $\boldsymbol{t}$-Step Denoiser}: This corresponds to $\boldsymbol{T}_{R_{\param},\lambda,\alpha}^t$ and is the denoiser optimized during training.
    It is natural to compare it to properly constrained PnP methods based on averaged deep denoisers as in \cite{bohra2021learning, NC2022}, which in general also do not correspond to minimizing an energy.
    \item {\bf Proximal Denoiser}: The learnt regularizer $R_{\param}$ is plugged into \eqref{eq:scalingfactor} with $\mat H = \mat I$, and the solution is computed using Algorithm~\ref{alg:FISTA} with small tolerance ($\num{e-6}$ for the relative change of norm between consecutive iterates).
    The parameters $\lambda$ and $\mu$ are tuned on the validation dataset with the coarse-to-fine method given in Appendix~\ref{ap:tuning}.
    This important step enables us to compensate for the gap between (i) gradient-step training and full minimization, and (ii) training and testing noise levels, if different.
\end{itemize}

\subsection{Denoising: Comparison with Other Methods}
Although not the final goal, image denoising yields valuable insights into the training of CRR-NNs.
It also enables us to compare CRR-NNs to the related methods given in Table~\ref{table:denoising_performance} on the standard BSD68 test set.\\
\begin{table}
\caption{Convex models and averaged denoisers tested on BSD68.}
\label{table:denoising_performance}
\centering
\begin{tabular}{lcc}
\toprule
  & $\sigma=5/255$ & $\sigma=25/255$ \\
\midrule
TV*\textsuperscript{,$\ddagger$} \cite{chambolle2004algorithm} & 36.41 & 27.48 \\
Higher-order MRFs*\textsuperscript{,$\ddagger$} \cite{CheRan2014} & NA & 28.04\\
$\mathrm{VN}^{1,t}\strut{}^\dagger$\cite{KKHP2017} & NA & 27.69\\
$\beta\mathrm{CNN}_{\sigma}\strut{}^\ddagger$ & 36.48 & 27.69\\
$\mathrm{D}_{\mathrm{ISTA}}\strut{}^\ddagger$ \cite{NC2022} & 36.54 & NA \\
GS-DnICNN\textsuperscript{$\dagger$}\cite{cohen2021has} & 36.85 & 27.76 \\
$\mathrm{D}_{\mathrm{ADMM}}\strut{}^\ddagger$\cite{NC2022} & 36.62 & NA\\
CRR-NN-ReLU ($t$-step)\textsuperscript{$\dagger$,$\ddagger$} & 35.50 & 26.75\\
CRR-NN ($t$-step)\textsuperscript{$\dagger$,$\ddagger$} & {\bf 36.97} & {\bf 28.12}\\
CRR-NN (proximal)*\textsuperscript{,$\ddagger$} & {\underline{36.96}} & {\underline{28.11}}\\
\bottomrule
\multicolumn{3}{l}{\footnotesize * Full minimization of a convex function}\\
\multicolumn{3}{l}{\footnotesize \textsuperscript{$\dagger$} Partial minimization of a convex function}\\
\multicolumn{3}{l}{\footnotesize \textsuperscript{$\ddagger$} Stable steps (averaged layers)}\\

\end{tabular}
\end{table}
Now, we briefly give the implementation details of the various frameworks.
CRR-NN-ReLU models are trained in the same way as CRR-NNs, but with ReLU activation functions (with learnable biases) instead of linear splines.
To emulate \cite{cohen2021has}, we train a DnICNN with the same architecture (ELU activations, $6$ layers, and $128$ channels per layer, $\num{745344}$ parameters) as a gradient step denoiser for $200$ epochs, separately for $\sigma\in\{5/255,25/255\}$, and refer to it as GS-DnICNN.
An averaged deep CNN denoiser $\beta\mathrm{CNN}_{\sigma}=\beta \V N + (1 - \beta)\mathbf{Id}$, with $\beta=0.5$, is trained on the same denoising task as the CRR-NNs with $\sigma\in\{5/255,25/255\}$.
Here, $\V N$ is chosen as a CNN with 9 layers, 64 channels, and PReLU activation functions, resulting in $\num{260225}$ learnable parameters.
The model is trained for 20 epochs with a batch size of 4 and a learning rate of $4 \cdot 10^{-5}$. To guarantee that $\V N$ is nonexpansive, the linear layers are spectral-normalized after each gradient step with the real-SN method \cite{ryu2019plug}, and the activations are constrained to be $1$-Lipschitz.
This CNN outperforms the averaged CNNs in \cite{bohra2021learning}.
Hence, it serves as a baseline for averaged deep CNNs.
The other reported frameworks do not provide public implementations.
Therefore, the numbers are taken from the corresponding papers.
Lastly, the TV denoising is performed with the algorithm proposed in \cite{chambolle2004algorithm}.
The results for all models are presented in Table~\ref{table:denoising_performance} and Figure~\ref{fig:denoising}.
\begin{itemize}
    \item {\bf $\boldsymbol{t}$-Step/Averaged Denoisers}: The CRR-NN-ReLU models perform poorly and confirms that ReLU is not well-suited to our setting.
    This limitation of ReLU was also observed experimentally in \cite{bohra2021learning} in the context of 1-Lipschitz denoisers. Our models improve over the gradient-step denoisers parameterized with ICNNs, even though the latter has many more parameters.
    The CRR-NN implementation improves over the special instance $\text{VN}^{1,t}$ of variational-network denoisers proposed in \cite{KKHP2017}, which also partially minimizes a convex cost.
    With a convex model similar to CRR-NNs (see Section~\ref{sc:implementation} for a discussion), it is shown that an increase in $t$ decreases the performance (reported as $\text{VN}_{24}^{1,t}$ in \cite[Figure 5]{KKHP2017}).
    The model $\text{VN}^{1,t}$ cannot compete with the proximal denoiser trained with bilevel optimization in \cite{CheRan2014}.
    By contrast, for $\sigma=25/255$ we obtain an improvement over $\text{VN}^{1,t}$ of $0.2$dB for $t=1$, and more than $0.6$dB as $t$ increases.
    Note that, in \cite{KKHP2017}, the layers of the $t$-step $\text{VN}^{1,t}$ denoiser are not guaranteed to be averaged.
    Our models also outperform the averaged $\beta\mathrm{CNN}_{\sigma}$ ($+0.5$dB for $\sigma=5$, $+0.4$dB for $\sigma=25/255$), and the two averaged denoisers $\mathrm{D}_{\mathrm{ISTA}}$ and $\mathrm{D}_{\mathrm{ADMM}}$ \cite{NC2022} ($+0.4$/$+0.3$dB for $\sigma=5/255$).
    In their simplest form, the latter are built with fixed linear layers (patch-based wavelet transforms) and learnable soft-thresholding activation functions.
    \item {\bf Proximal Denoisers}: Our models yield slight improvements over the higher-order Markov random field (MRF) model in the pioneering work \cite{CheRan2014} ($28.04$dB vs $28.11$dB for $\sigma=25/255$).
    With a similar architecture---but with fixed smoothed absolute value $\psi_i$---the latter approach involves a computationally intensive bilevel optimization with second-order solvers.
    Here, we show that a few gradient steps for training already suffice to be competitive.
    This leads to ultrafast training and bridges the gap between higher-order MRF models and VN denoisers.
    Lastly, we remark that our proximal denoisers are robust to a mismatch in the training and testing noise levels.
\end{itemize}

\begin{figure}[t]
    \centering
    \includegraphics[width=0.48\textwidth]{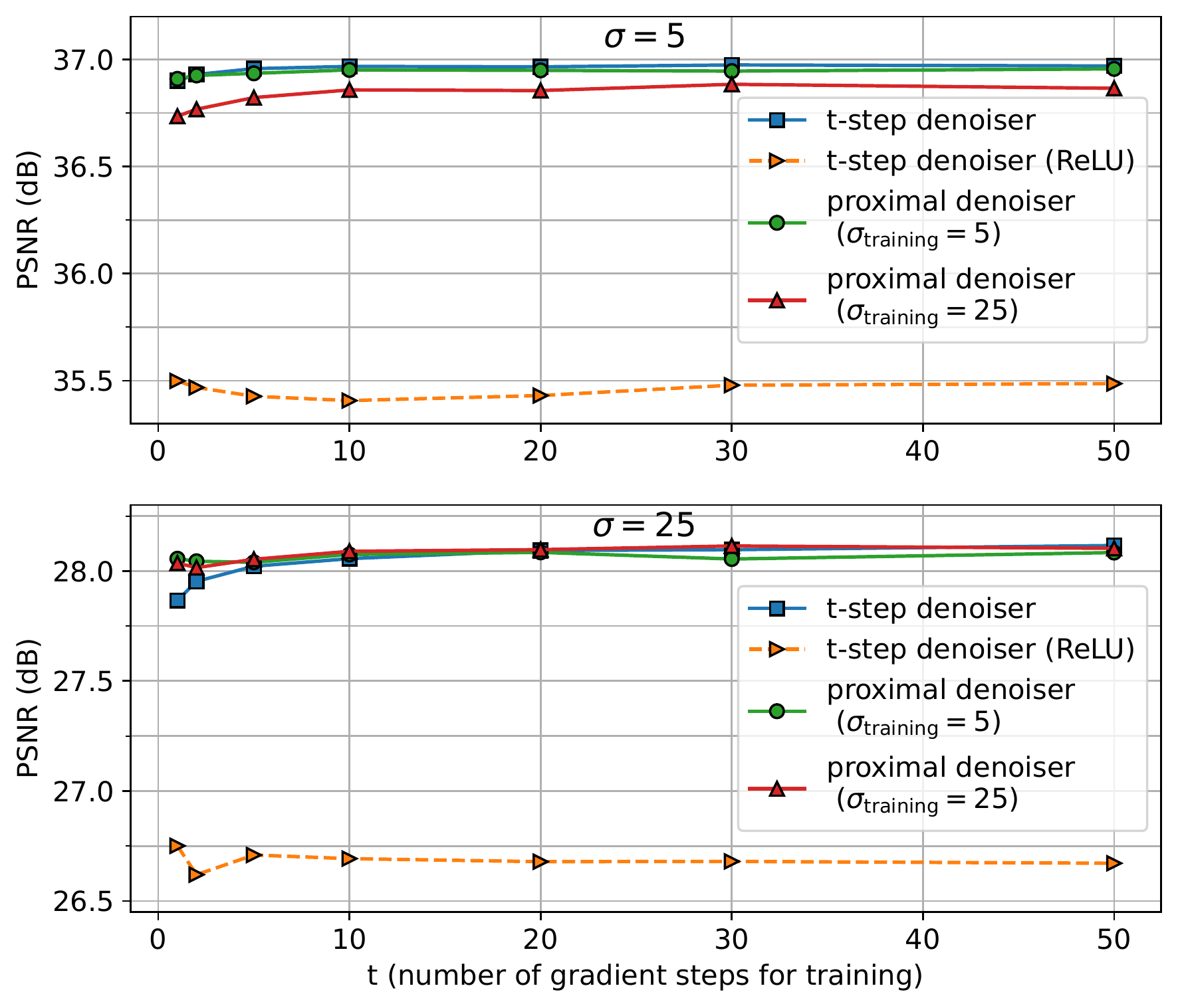}
    \caption{Test denoising performance of CRR-NNs for noise level $\sigma = 5/255$ and $\sigma = 25/255$ versus the number of gradient steps used for training, the denoiser type ($t$-step vs proximal), and the noise level used for training.}
    \label{fig:denoising}
\end{figure}

\subsection{Biomedical Image Reconstruction}
\label{subsec:applications}
The six CRR-NNs trained on denoising with $t\in\{1, 10, 50\}$ and $\sigma\in\{5/255, 25/255\}$ are now used to solve the following two ill-posed inverse problems.
\paragraph{MRI}
The ground-truth images for our MRI experiments are proton-density weighted knee MR images from the fastMRI dataset \cite{zbontar2018fastMRI} with fat suppression (PDFS) and without fat suppresion (PD). They are generated from the fully-sampled k-space data. For each of the two categories (PDFS and PD), we create validation and test sets consisting of 10 and 50 images, respectively, where every image is normalized to have a maximum value of one. To gauge the performance of CRR-NNs in various regimes, we experiment with single-coil and multi-coil setups with several acceleration factors. In the single-coil setup, we simulate the measurements by masking the Fourier transform of the ground-truth image. In the multi-coil case, we consider 15 coils, and the measurements are simulated by subsampling the Fourier transforms of the multiplication of the ground-truth images with 15 complex-valued sensitivity maps (these were estimated from the raw k-space data using the ESPIRiT algorithm \cite{Uecker2014-uv} available in the BART toolbox \cite{uecker2013software}). For both cases, the subsampling in the Fourier domain is performed with a Cartesian mask that is specified by two parameters: the acceleration $M_{\text{acc}}\in \{2, 4, 8\}$ and the center fraction $M_{\text{cf}}=0.32/M_{\text{acc}}$. A fraction of $M_{\text{cf}}$ columns in the center of the k-space (low frequencies) is kept, while columns in the other region of the k-space are uniformly sampled so that the expected proportion of selected columns is $ 1/M_{\text{acc}}$. In addition, Gaussian noise with standard deviation $\sigma_{\mathbf{n}} = \num{2e-3}$ is added to the real and imaginary parts of the measurements. The PSNR and SSIM values for each method are computed on the $(320\times 320)$ centered ROI.
\begin{table}
\centering
\caption{Single-coil MRI.}%

\label{table:reconstruction_performance_mri_single}
\setlength\tabcolsep{2pt}
\begin{tabular}{lcccccccc}
\toprule
& \multicolumn{4}{c}{2-fold} & \multicolumn{4}{c}{4-fold} \\
& \multicolumn{2}{c}{PSNR} & \multicolumn{2}{c}{SSIM} & \multicolumn{2}{c}{PSNR} & \multicolumn{2}{c}{SSIM}  \\
& PD & PDFS & PD & PDFS & PD & PDFS & PD & PDFS\\
\midrule
Zero-fill & 33.32 & 34.49 & 0.871 & 0.872 & 27.40 & 29.68 & 0.729 & 0.745  \\
TV & 39.22 & 37.73 & 0.947 & 0.917 & 32.44 & 32.67 & 0.833 & 0.781 \\
PnP-$\beta$CNN & 38.77 & 37.89 & 0.943 & 0.924 & 31.37 & 31.82 & 0.832 & 0.797   \\
CRR-NN & {\bf 40.95} & \underline{38.91} & {\bf 0.961} & \underline{0.934} & \underline{33.99} & \underline{33.75} & \underline{0.880} & \underline{0.831} \\
\midrule
PnP-DnCNN \cite{ryu2019plug} & \underline{40.52} & {\bf 39.02} & \underline{0.956} & {\bf 0.935} & {\bf 35.24} & {\bf 34.63} & {\bf 0.884} & {\bf 0.840} \\
\bottomrule
\end{tabular}
\end{table}

\begin{table}
\centering
\caption{CRR-NN: Single-coil MRI versus training setup.}%

\label{table:reconstruction_performance_mri_single_setup}
\setlength\tabcolsep{2pt}
\begin{tabular}{lrrcccccccc}
\toprule
& & & \multicolumn{4}{c}{2-fold} & \multicolumn{4}{c}{4-fold} \\
& & & \multicolumn{2}{c}{PSNR} & \multicolumn{2}{c}{SSIM} & \multicolumn{2}{c}{PSNR} & \multicolumn{2}{c}{SSIM}  \\
image & $\sigma_{\mathrm{train}}$ & t & PD & PDFS & PD & PDFS & PD & PDFS & PD & PDFS\\
\midrule
BSD & 5/255  & 1  & 40.55 & 38.71 & 0.959 & 0.932 & 33.32 & 33.37 & 0.866 & 0.819  \\
BSD & 5/255  & 10 & 40.52 & 38.69 & 0.959 & 0.932 & 33.30 & 33.36 & 0.865 & 0.817 \\
BSD & 5/255  & 50 & 40.50 & 38.67 & 0.958 & 0.931 & 33.29 & 33.32 & 0.865 & 0.816   \\
BSD & 25/255 & 1  & 40.75 & 38.84 & 0.960 & 0.934 & 33.62 & 33.60 & 0.875 & 0.828 \\
BSD & 25/255 & 10 & 40.78 & 38.81 & 0.960 & 0.933 & 33.63 & 33.59 & 0.875 & 0.826 \\
BSD & 25/255 & 50 & 40.71 & 38.77 & 0.960 & 0.932 & 33.57 & 33.54 & 0.872 & 0.824 \\
MRI & 5/255  & 10 & 40.95 & 38.91  & 0.961 & 0.934 & 33.99  & 33.75 & 0.880 & 0.831 \\
MRI & 25/255  & 10 & 40.61 & 38.73 & 0.959 & 0.932 & 33.93 & 33.71 & 0.878 & 0.830 \\
\bottomrule
\end{tabular}
\end{table}

\begin{table}
\centering
\caption{Multi-coil MRI.}%

\label{table:reconstruction_performance_mri_multicoil}
\setlength\tabcolsep{2pt}
\begin{tabular}{lcccccccc}
\toprule
& \multicolumn{4}{c}{4-fold} & \multicolumn{4}{c}{8-fold} \\
& \multicolumn{2}{c}{PSNR} & \multicolumn{2}{c}{SSIM} & \multicolumn{2}{c}{PSNR} & \multicolumn{2}{c}{SSIM}  \\
& PD & PDFS & PD & PDFS & PD & PDFS & PD & PDFS\\
\midrule
$\vec H^T \vec y$ & 27.71 & 29.94 & 0.751 & 0.759 & 23.80 & 27.19 & 0.648 & 0.681 \\
TV & 38.06 & 37.31 & 0.935 & 0.914 & 32.77 & 33.38 & 0.850 & 0.824 \\
PnP-$\beta$CNN & 37.88 & 37.48 & 0.934 & 0.919 & 32.52 & 33.30 & 0.849 & 0.832  \\
CRR-NN & \underline{39.54} & \underline{38.29} & {\bf 0.950} & \underline{0.927} & \underline{34.29} & \underline{34.50} & {\bf 0.881} & \underline{0.852} \\
\midrule

PnP-DnCNN \cite{ryu2019plug} & {\bf 39.55} & {\bf 38.52} & \underline{0.947} & {\bf 0.929} & {\bf 35.11} & {\bf 35.14} & {\bf 0.881} & {\bf 0.858}\\
\bottomrule
\end{tabular}
\end{table}

\begin{table}
\centering
\caption{CRR-NN: Multi-coil MRI versus training setup.}%

\label{table:reconstruction_performance_mri_multicoil_setup}
\setlength\tabcolsep{2pt}
\begin{tabular}{lrrcccccccc}
\toprule
& & & \multicolumn{4}{c}{4-fold} & \multicolumn{4}{c}{8-fold} \\
& & & \multicolumn{2}{c}{PSNR} & \multicolumn{2}{c}{SSIM} & \multicolumn{2}{c}{PSNR} & \multicolumn{2}{c}{SSIM}  \\
image & $\sigma_{\mathrm{train}}$ & t & PD & PDFS & PD & PDFS & PD & PDFS & PD & PDFS\\
\midrule
BSD & 5/255  & 1  & 39.15 & 38.09 & 0.947 & 0.925 & 33.82 & 34.22 & 0.873 & 0.846  \\
BSD & 5/255  & 10 & 39.14 & 38.08 & 0.946 & 0.925 & 33.82 & 34.20 & 0.873 & 0.845 \\
BSD & 5/255  & 50 & 39.14 & 38.05 & 0.946 & 0.924 & 33.78 & 34.16 & 0.872 & 0.844   \\
BSD & 25/255 & 1  & 39.34 & 38.21 & 0.948 & 0.926 & 34.02 & 34.35 & 0.876 & 0.849 \\
BSD & 25/255 & 10 & 39.33 & 38.19 & 0.948 & 0.926 & 34.01 & 34.34 & 0.876 & 0.848 \\
BSD & 25/255 & 50 & 39.29 & 38.15 & 0.948  & 0.926 & 33.96 & 34.29 & 0.876 & 0.847 \\
MRI & 5/255  & 10 & 39.54 & 38.29 & 0.950  & 0.927 & 34.29 & 34.50 & 0.881 & 0.852 \\
MRI & 25/255  & 10 & 39.33 & 38.14 & 0.947 & 0.925 & 34.22 & 34.40 & 0.878 & 0.849 \\
\bottomrule
\end{tabular}
\end{table}

\paragraph{CT}
To provide a fair comparison with the ACR method, we now target the CT experiment proposed in \cite{MukDit2021}.
The data consist of human abdominal CT scans for 10 patients provided by Mayo Clinic for the low-dose CT Grand Challenge \cite{M2016}.
The validation set consists of 6 images taken uniformly from the first patient of the training set from \cite{MukDit2021}.
We use the same test set as \cite{MukDit2021}, more precisely, 128 slices with size $(512\times 512)$ that correspond to one patient.
The projections of the data are simulated using a parallel-beam acquisition geometry with 200 angles and 400 detectors.
Lastly, Gaussian noise with standard deviation $\sigma_{\mathbf{n}}\in\{ 0.5, 1, 2\}$ is added to the measurements.
\begin{table}
\centering
\caption{CT.}%

\label{table:reconstruction_performance_ct}
\setlength\tabcolsep{4.5pt}
\begin{tabular}{lcccccc}
\toprule
& \multicolumn{2}{c}{$\sigma_{\vec n}$=0.5} & \multicolumn{2}{c}{$\sigma_{\vec n}$=1} & \multicolumn{2}{c}{$\sigma_{\vec n}$=2} \\
& PSNR & SSIM & PSNR & SSIM & PSNR & SSIM  \\
\midrule
FBP & 32.14 & 0.697 & 27.05 & 0.432 & 21.29 & 0.204 \\
TV & 36.38 & 0.936 & 34.11 &  0.906 & 31.57 & 0.863 \\
PnP-$\beta$CNN & 37.19 & 0.920 & 34.11 & 0.873 &  30.93 &  0.804 \\
ACR \cite{MukDit2021,MukSch2021} & 38.06 & \underline{0.943} & 35.12  &  0.911 &  32.17 & 0.868\\
CRR-NN & \bf{39.30} & \bf{0.947} & \underline{36.29}  &  \underline{0.916} &  \underline{33.16} & \underline{0.878}\\
\midrule
PnP-DnCNN \cite{ryu2019plug} & \underline{38.93} &  0.941 & \bf{36.49} & \bf{0.921} & \bf{33.52} & \bf{0.897} \\
\bottomrule
\end{tabular}
\end{table}

\begin{table}
\centering
\caption{CRR-NN: CT versus training setup.}%

\label{table:reconstruction_performance_ct_setup}
\setlength\tabcolsep{3pt}
\begin{tabular}{lrrccccccc}
\toprule
& & & \multicolumn{2}{c}{$\sigma_{\vec n}$=0.5} & \multicolumn{2}{c}{$\sigma_{\vec n}$=1} & \multicolumn{2}{c}{$\sigma_{\vec n}$=2} \\

image & $\sigma_{\mathrm{train}}$ & t & PSNR & SSIM & PSNR & SSIM & PSNR & SSIM \\
\midrule
BSD & 5/255  & 1  & 38.84 & 0.943 & 35.70 & 0.907 & 32.48 & 0.860\\ 
BSD & 5/255  & 10 & 38.90 & 0.943 & 35.73 & 0.908 & 32.49 & 0.860 \\ 
BSD & 5/255  & 50 & 38.82 & 0.940 & 35.64 & 0.904 &  32.47& 0.855 \\ 
BSD & 25/255 & 1  & 39.01 & 0.945 & 35.91 & 0.913 & 32.72 & 0.867 \\ 
BSD & 25/255 & 10 & 39.07 & 0.945 & 35.95 & 0.911 & 32.71 & 0.867 \\ 
BSD & 25/255 & 50 & 39.04 & 0.944 & 35.89 & 0.912 & 32.71 & 0.860 \\ 
CT & 5/255   & 10 & 39.30 & 0.947 & 36.29 & 0.916 & 33.15 & 0.873 \\ 
CT & 25/255  & 10 & 38.89 & 0.945 & 36.11 & 0.917 & 33.16 & 0.878 \\
\bottomrule
\end{tabular}
\end{table}

\begin{figure*}[t]
    \centering
    \includegraphics[width=0.83\textwidth]{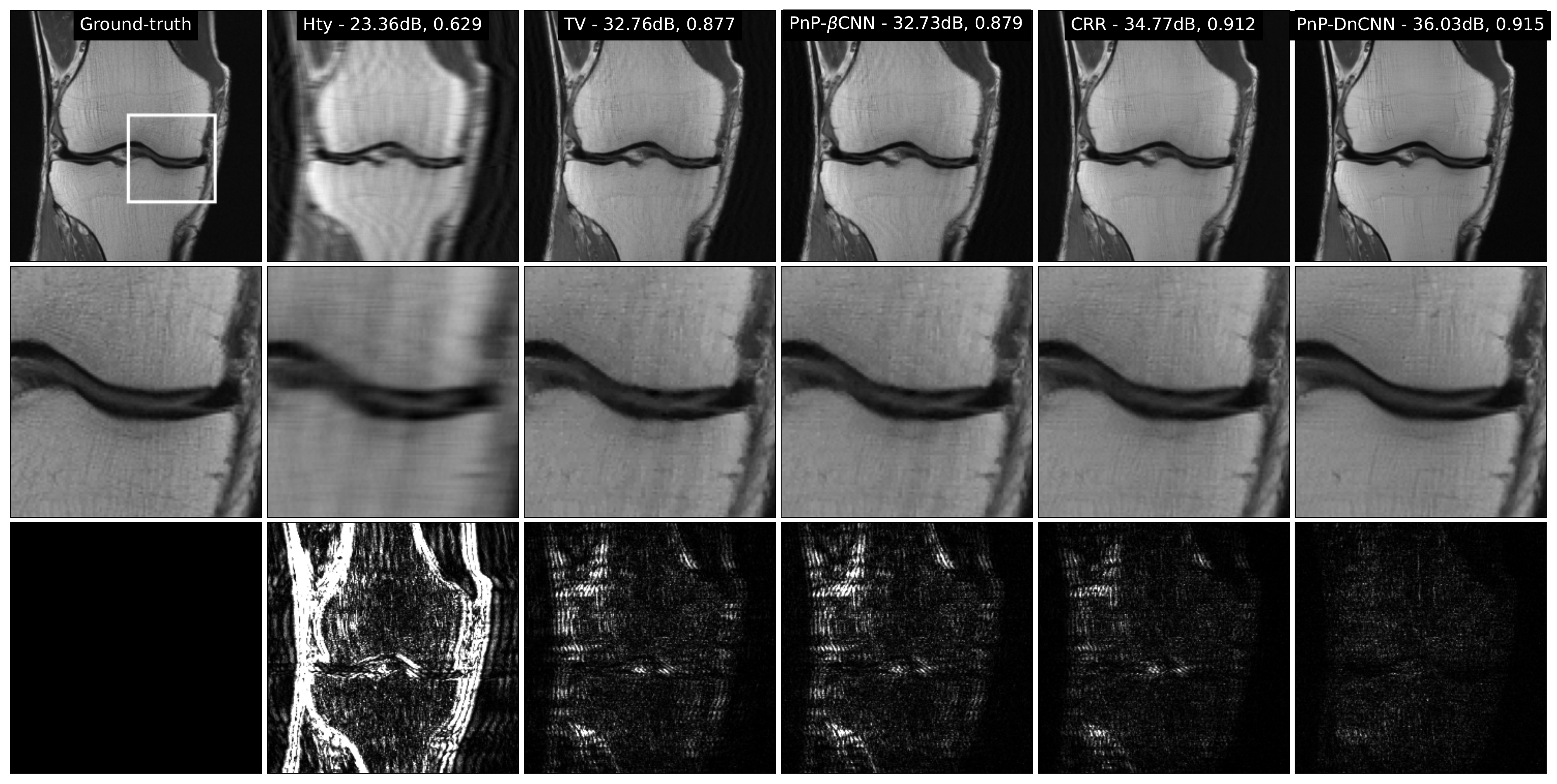}
    \caption{Reconstructed images for the 4-fold accelerated multi-coil MRI experiment. The reported metrics are PSNR and SSIM. The last row shows the squared differences between the reconstructions and the ground-truth image.}
    \label{fig:MRIreconstructions}
\end{figure*}

\begin{figure*}[t]
    \centering
    \includegraphics[width=\textwidth]{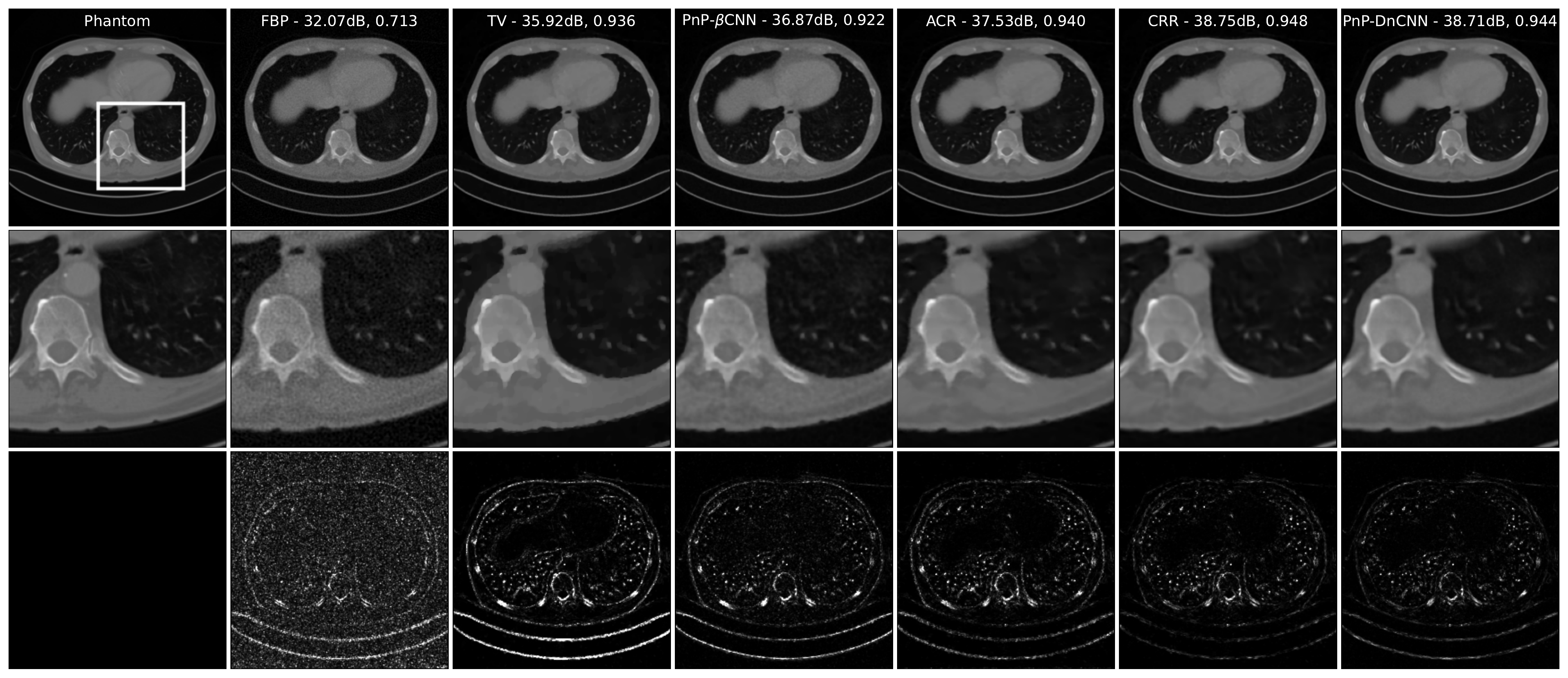}
    \caption{Reconstructed images for the CT experiment with $\sigma_{\vec n}=0.5$. The reported metrics are PSNR and SSIM. The last row shows the squared differences between the reconstructions and the ground-truth image.}
    \label{fig:CTreconstructions}
\end{figure*}

\paragraph{Reconstruction Frameworks}
A reconstruction with isotropic TV regularization is computed with FISTA \cite{beck2009fast}, in which $\prox_R$ is computed as in \cite{BT2009} to enforce positivity.
We also consider reconstructions obtained with the PnP method with (i) provably averaged denoisers $\beta\mathrm{CNN}_{\sigma}$ ($\sigma=5,25$); and (ii) the popular pertained DnCNNs \cite{ryu2019plug} ($\sigma=5, 15, 40$). The latter are residual denoisers with 1-Lipschitz convolutional layers and batch normalization modules, which yield a non-averaged denoiser with no convergence guarantees for ill-posed problems.
To adapt the strength of the denoisers, in addition to the training noise level, we use relaxed denoisers $\V D_\gamma=\gamma \V D + (1-\gamma)\mathbf{Id}$ for all denoisers $\V D$, where $\gamma\in(0,1]$ is tuned along with the stepsize $\alpha$ given in \eqref{eq:PnPFBS}.
We only report the performance of the best-performing setting.
The ACR framework \cite{MukDit2021,MukSch2021} yields a convex regularizer for \eqref{eq:VarProb} that is specifically designed to the described CT problem.
To be consistent with \cite{MukDit2021, MukSch2021}, we apply 400 iterations of gradient descent, even though the objective is nonsmooth, and tune the stepsize and $\lambda$.
The results are consistent with those reported in \cite{MukDit2021, MukSch2021}.

To assess the dependence of CRR-NNs on the image domain, we also train models for Gaussian denoising of CT and MRI images ($t=10$, $\sigma\in\{5/255, 25/255\}$). The training procedure is the same as for BSD image denoising, but a larger kernel size of 11 was required to saturate the performance. The learnt filters and activations are included in the Supplementary Material.

The hyperparameters for all these methods are tuned to maximize the average PSNR over the validation set with the coarse-to-fine method given in Appendix~\ref{ap:tuning}.
\paragraph{Results and Discussion}
For each modality, a reconstruction example is given for each framework in Figures \ref{fig:MRIreconstructions} and \ref{fig:CTreconstructions}, and additional illustrations are given in the Supplementary Material. The PSNR and SSIM values for the test set given in Tables \ref{table:reconstruction_performance_mri_single}, \ref{table:reconstruction_performance_mri_multicoil}, and \ref{table:reconstruction_performance_ct} attest that CRR-NNs consistently outperform the other frameworks with comparable guarantees. It can be seen from Tables \ref{table:reconstruction_performance_mri_single_setup}, \ref{table:reconstruction_performance_mri_multicoil_setup}, and \ref{table:reconstruction_performance_ct_setup} that the improvements hold for all setups explored to trained CRR-NNs. The training of CRR-NNs on the target image domain allows for an additional small performance boost. The performances of CRR-NNs are close to the ones of PnP-DnCNN, which has however no guarantees and little interpretability. PnP-DnCNN typically yields artifact-free reconstructions but is more prone to over-smoothing (Figure \ref{fig:MRIreconstructions}) or even to exaggeration of some details in rare cases (see Figures in the Supplementary Material). Lastly, observe that the properly constrained PnP-$\beta$CNN is not always competitive with TV. This confirms the difficulty of training provably 1-Lipchitz CNN, which is also reported for MRI image reconstruction in \cite{chand2023multiscale}. Convergence curves for CRR-NNs can be found in the Supplementary Material.

\subsection{Under the Hood of the Learnt Regularizers}
\label{subsec:underthehood}
The filters and activation functions for learnt CRR-NNs with $\sigma \in \{5/255,25/255\}$ and $t=5$ are shown in Figures~\ref{fig:filteractivationsig5} and~\ref{fig:filteractivationsig25}.
\subsubsection{Filters} The impulse responses of the filters vary in orientation and frequency response.
This indicates that the CRR-NN decouples the frequency components of patches.
The learnt kernels typically come in groups that are reminiscent of 2D steerable filters \cite{freeman1991design,UC2013}.
Interestingly, their support is wider when the denoising task is carried out for $\sigma=25/255$ than for $\sigma=5/255$.
\subsubsection{Activation Functions}
The linear splines converge to simple functions throughout the training.
The regularization~\eqref{eq:TrainProbReg} leads to even simpler ones without a compromise in performance.
Most of them end up with 3 linear regions, with their shape being reminiscent of the clipping function $\mathrm{Clip}(x)=\mathrm{sign}(x)\min(|x|,1)$.
The learnt regularizer is closely related to $\ell_1$-norm based regularization as many of the learnt convex profiles $\psi_i$ resemble some smoothed version of the absolute-value function.
\subsubsection{Pruning CRR-NNs}
Since the NN has a simple architecture, it can be efficiently pruned before inference by removal of the filters associated with almost-vanishing activation functions. This yields models with typically between $\num{3000}$ and $\num{5000}$ parameters and offers a clear advantage over deep models, which can usually not be pruned efficiently.
\subsubsection{A Signal-Processing Interpretation}Given that the gradient-step operator $\vec x \mapsto (\vec x - \alpha \mat W^T \boldsymbol{\sigma}(\mat W \vec x))$ of the learnt regularizer is expected to remove some noise from $\vec x$, the 1-hidden-layer CNN $\mat W^T \boldsymbol{\sigma}(\mat W \vec \cdot)$ is expected to extract noise.
The response of $\vec x$ to the learnt filters forms the high-dimensional representation $\mat W \vec x$ of $\vec x$.
The clipping function preserves the small responses to the filters, while it cuts the large ones.
Hence, the estimated noise $\mat W^T \boldsymbol{\sigma}(\mat W \vec x)$ is reconstructed by essentially removing the components of $\vec x$ that exhibit a significant correlation with the kernels of the filters.
All in all, the learning of the activation functions leads closely to wavelet- or framelet-like denoising.
Indeed, the proximal operator of $\vec{x}\mapsto \|\mathrm{DWT}(\vec x)\|_1$ is given by
\begin{align}
\mathrm{prox}_{\|\mathrm{DWT}(\vec \cdot)\|_1}(\vec x) &= \mathrm{IDWT}(\mathrm{soft}(\mathrm{DWT}(\vec x)))\notag\\
&= \vec x - \mathrm{IDWT}(\mathrm{clip}(\mathrm{DWT}(\vec x))),
\end{align}
where $\mathrm{soft}(\cdot)$ is the soft-thresholding function, $\mathrm{DWT}$ and $\mathrm{IDWT}$ are the orthogonal discrete wavelet transform and its inverse, respectively. The equivalent formulation with the clipping function follows from $\mathrm{IDWT}(\mathrm{DWT}(\vec x)) =\vec x$ and $\mathrm{soft}(\vec x) = (\vec x - \mathrm{clip}(\vec x))$. The soft-thresholding function is used for direct denoising while the clipping function is tailored to residual denoising.
Note that the given analogy is, however, limited since the learnt filters are not orthonormal ($\mat W^T \mat W \neq \mat I$).
\begin{figure}[t]
    \centering
    \includegraphics[width=0.48\textwidth]{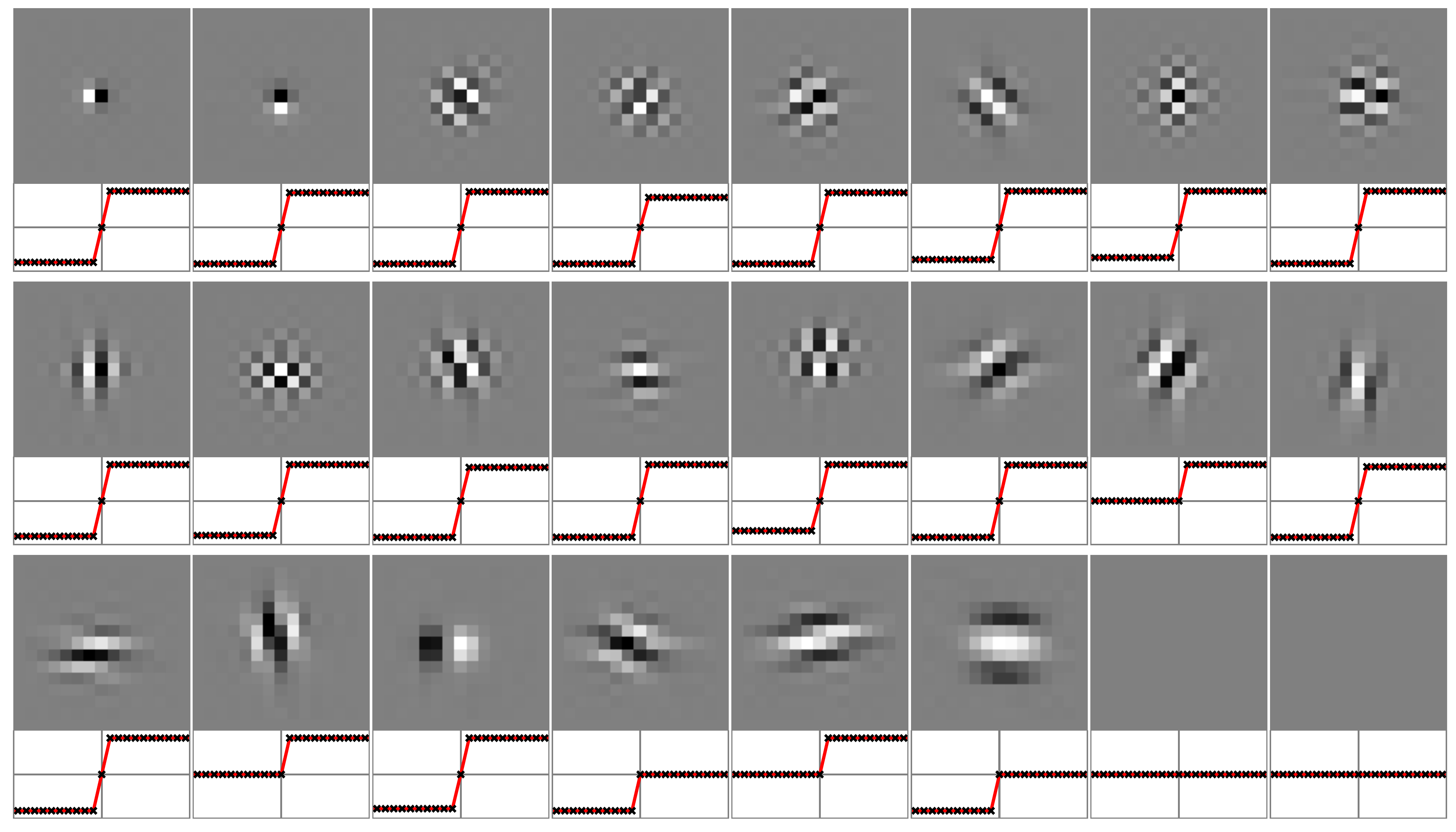}
    \caption{Impulse response of the filters and activation functions of the CRR-NN trained with $\sigma=5$.
    The crosses indicate the knots of the splines.
    For the 8 missing filters, the associated activation functions were numerically identically zero.}
    \label{fig:filteractivationsig5}
\end{figure}
\begin{figure}[t]
    \centering
    \includegraphics[width=0.48\textwidth]{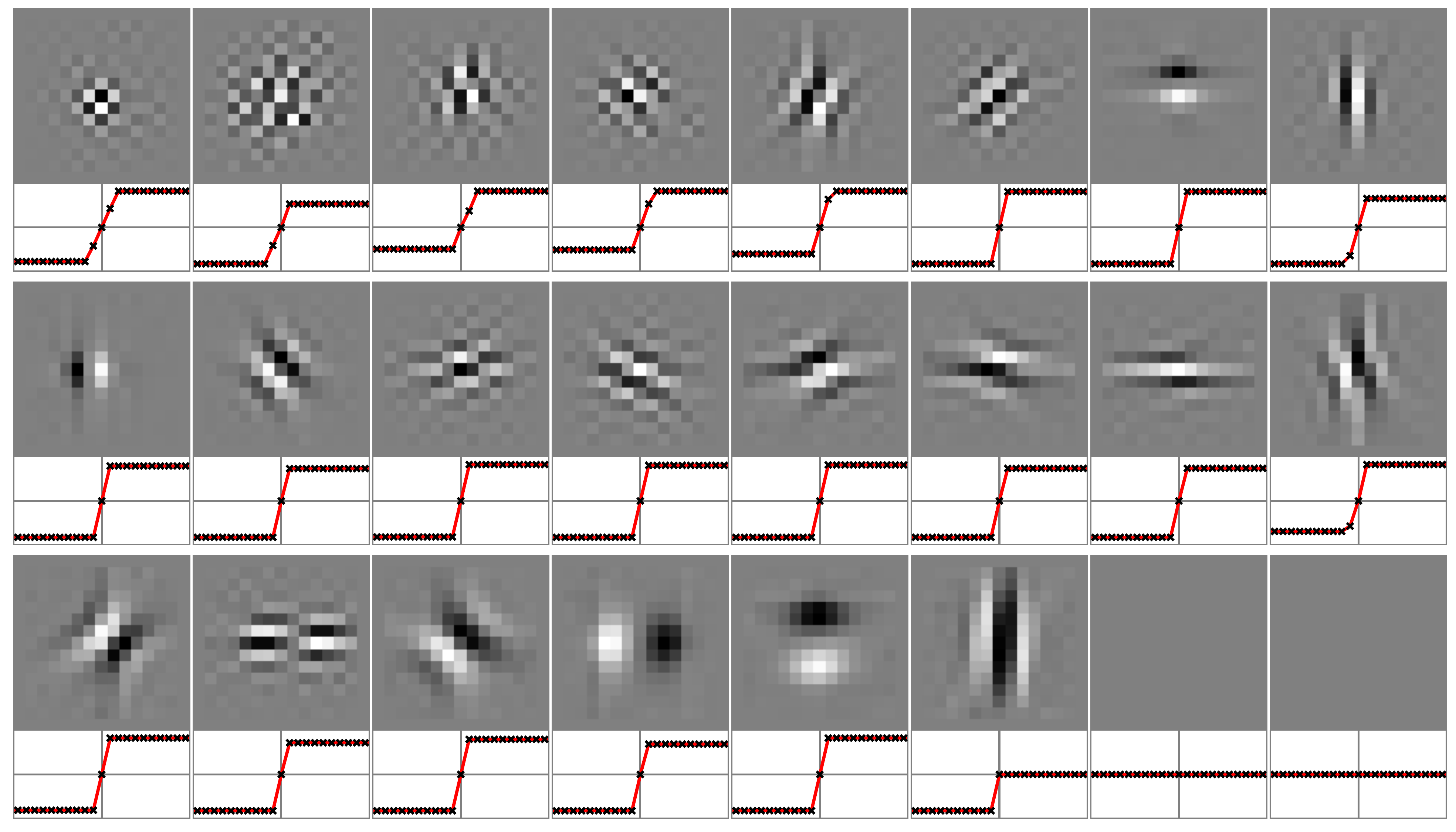}
    \caption{Impulse response of the filters and activation functions of the CRR-NN trained with $\sigma=25/255$.}
    \label{fig:filteractivationsig25}
\end{figure}

\subsubsection{Role of the Scaling Factor}
To clarify the role of the scaling factor $\mu$ introduced in~\eqref{eq:scalingfactor}, we investigate a toy problem on the space of one-dimensional signals.
Since these can be interpreted as images varying along a single direction, a signal regularizer $R_1$ can be obtained from $R_{\param}$ by replacing the 2D convolutional filters with 1D convolutional filters whose kernels are the ones of $R_{\param}$ summed along a direction.
Next, we seek a compactly supported signal with fixed mass that has minimum regularization cost, as in
\begin{equation}
\label{eq:1dproblem}
    \hat{\vec{c}} = \argmin \limits_{\vec c \in\R^d} R_1(\mu \vec c) \,\,\, \text{s.t.}\,\,\,  \begin{cases}\vec 1^T \vec c = 1,\\
\vec c_k = 0, & \forall k\not\in [k_1,k_2].\end{cases}
\end{equation}
The solutions for various values of $\mu$ are shown in Figure \ref{fig:canonicalfunctions}.
Small values of $\mu$ promote smooth functions in a way reminiscent of the Tikhonov regularizer applied to finite differences. Large values of $\mu$ promote functions with constant portions and, conjointly, allows for sharp jumps, which is reminiscent of the TV regularizer. This reasoning is in agreement with the shape of the activation functions shown in Figures~\ref{fig:filteractivationsig5} and~\ref{fig:filteractivationsig25}.
Indeed, an increase in $\mu$ allows one to enlarge the region where the regularizer has constant gradients, while a decrease of $\mu$ allows one to enlarge the region where the regularizer has linear gradients.

\begin{figure}[t]
    \centering
    \includegraphics[width=0.25\textwidth]{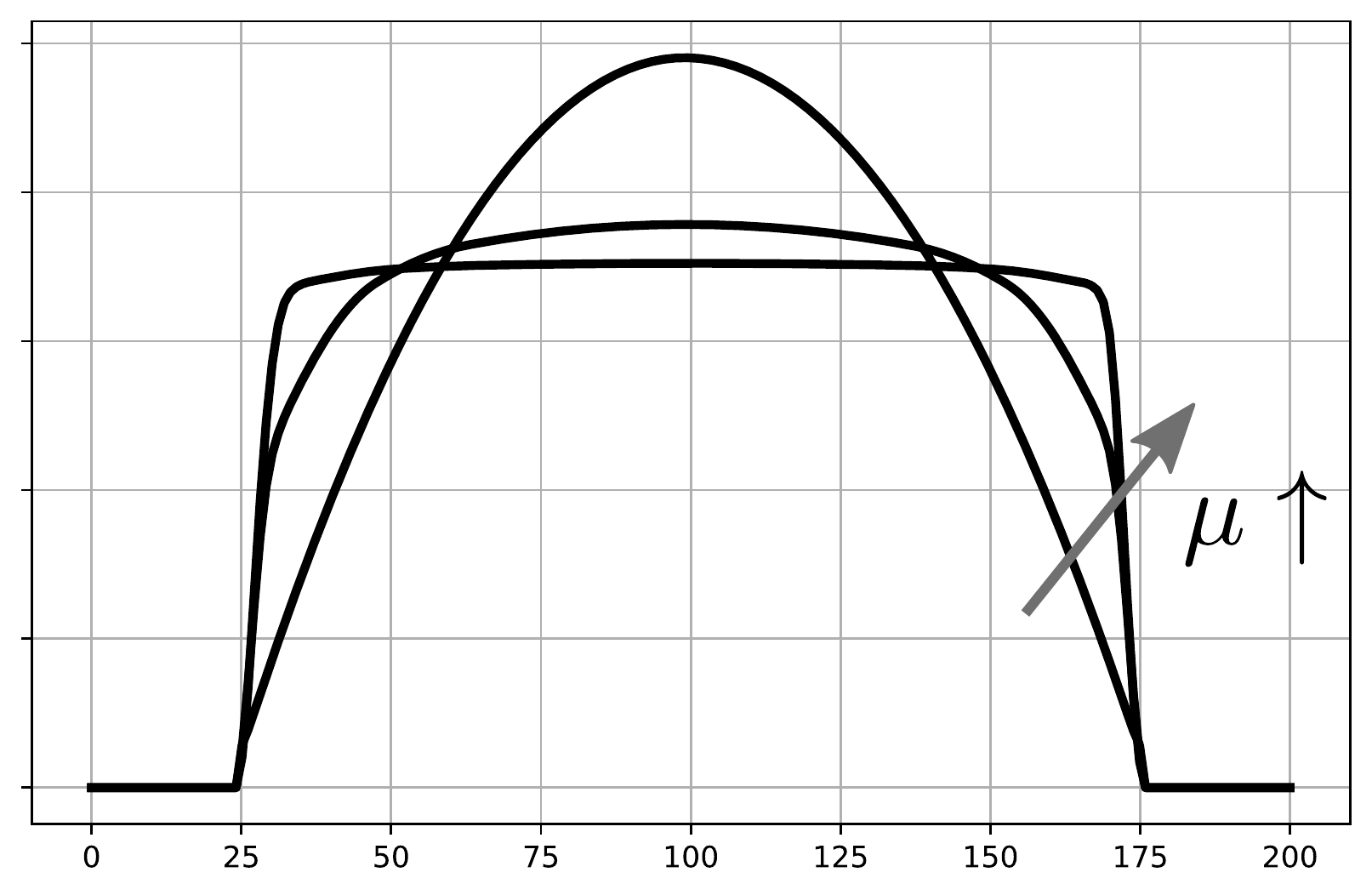}
    \caption{Solutions of the one-dimensional problem~\eqref{eq:1dproblem} for increasing values of $\mu$. The plotted functions are supported in $[25, 175]$ and minimize the learnt regularizer given a unit sum of their values.}
    \label{fig:canonicalfunctions}
\end{figure}

\section{Conclusion}
We have proposed a framework to learn universal convex-ridge regularizers with adaptive profiles. When applied to inverse problems, it is competitive with those recent deep-learning approaches that also prioritize the reliability of the method. Not only CRR-NNs are faster to train, but they also offer improvements in image quality. The findings raise the question of whether shallow models such as CRR-NNs, despite their small number of parameters, already offer optimal performance among methods that rely either on a learnable convex regularizer or on the PnP framework with a provably averaged denoiser.
In the future, CRR-NNs could be fine-tuned on specific modalities via the use of $\mat H$ for training. This could further improve the reconstruction quality, as observed when shifting from PnP to deep unrolled algorithms while maintaining the guarantees.

\section*{Acknowledgments}
The research leading to these results was supported by the European Research Council (ERC) under European Union’s Horizon 2020 (H2020), Grant Agreement - Project No 101020573 FunLearn and by the Swiss National Science Foundation, Grant 200020 184646/1. The authors are thankful to Dimitris Perdios for helpful discussions.
\bibliographystyle{IEEEtran}
\bibliography{references, new_ref}
\appendix

\subsection{Hyperparameter Tuning}
\label{ap:tuning}
The parameters $\lambda$ and $\mu$ used in \eqref{eq:scalingfactor} can be tuned with a coarse-to-fine approach. Given the performance on the $3\times 3$ grid $\{(\gamma_\lambda)^{-1}\lambda, \lambda, \gamma_\lambda\lambda\}\times \{(\gamma_\mu)^{-1}\mu, \mu, \gamma_\mu \mu \}$, we identify the best values $\lambda^*$ and $\mu^*$ on this subset and move on to the next iteration as follows:
\begin{itemize}
    \item if $\lambda^* = \lambda$, we refine the search grid by reducing $\gamma_\mu$ to $(\gamma_\mu)^{\zeta}$, $\zeta<1$;
    \item otherwise, $\lambda$ is updated to $\lambda^*$.
\end{itemize}
A similar update is performed for the scaling parameter.
The search is terminated when both $\gamma_\lambda$ and $\gamma_\mu$ are smaller than a threshold, typically, $1.01$. In practice, we initialized $\gamma_\lambda=\gamma_\mu=4$ and set $\zeta=0.5$. The method usually requires between $50$ and $100$ evaluations on tuples $(\lambda, \mu)$ on the validation set before it terminates. The proposed approach is predicated on the observation that the optimization landscape in the $(\lambda,\mu)$ domain is typically well-behaved. The same principles apply to tune a single hyperparameter, as found in the TV and the PnP-$\beta$CNN methods.
Let us remark that the performances were found to change only slowly with the scaling parameter $\mu$ for the MRI and CT experiments.
Hence, in practice, it is enough to tune $\mu$ very coarsely.
\newpage
\section*{Supplementary Material}

\begin{figure}[!ht]
    \centering
    \includegraphics[width=0.41\textwidth]{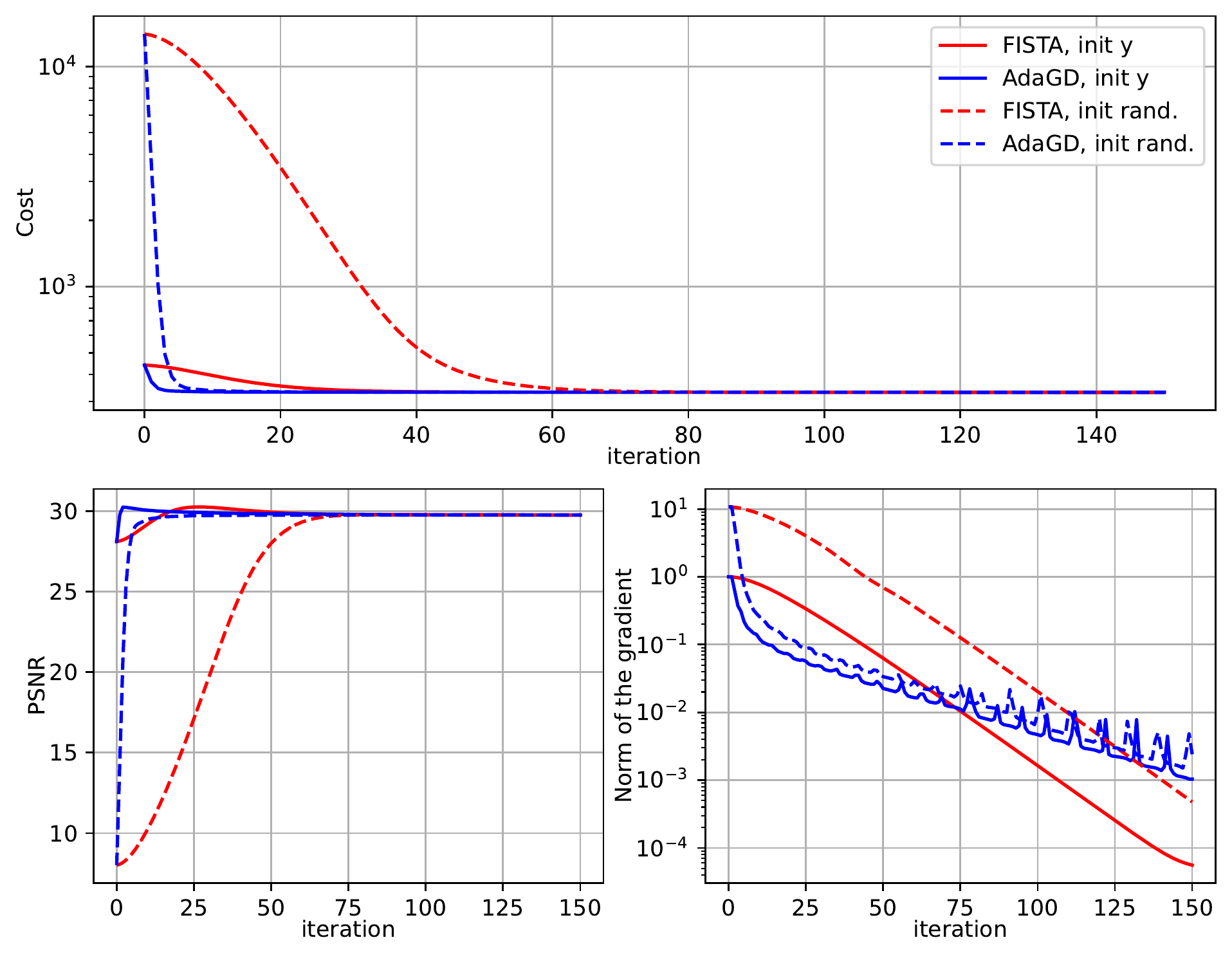}
    \caption{Example of convergence curves (denoising).}
    \label{fig:convergencedenoising}
\end{figure}
\begin{figure}[!ht]
    \centering
    \includegraphics[width=0.41\textwidth]{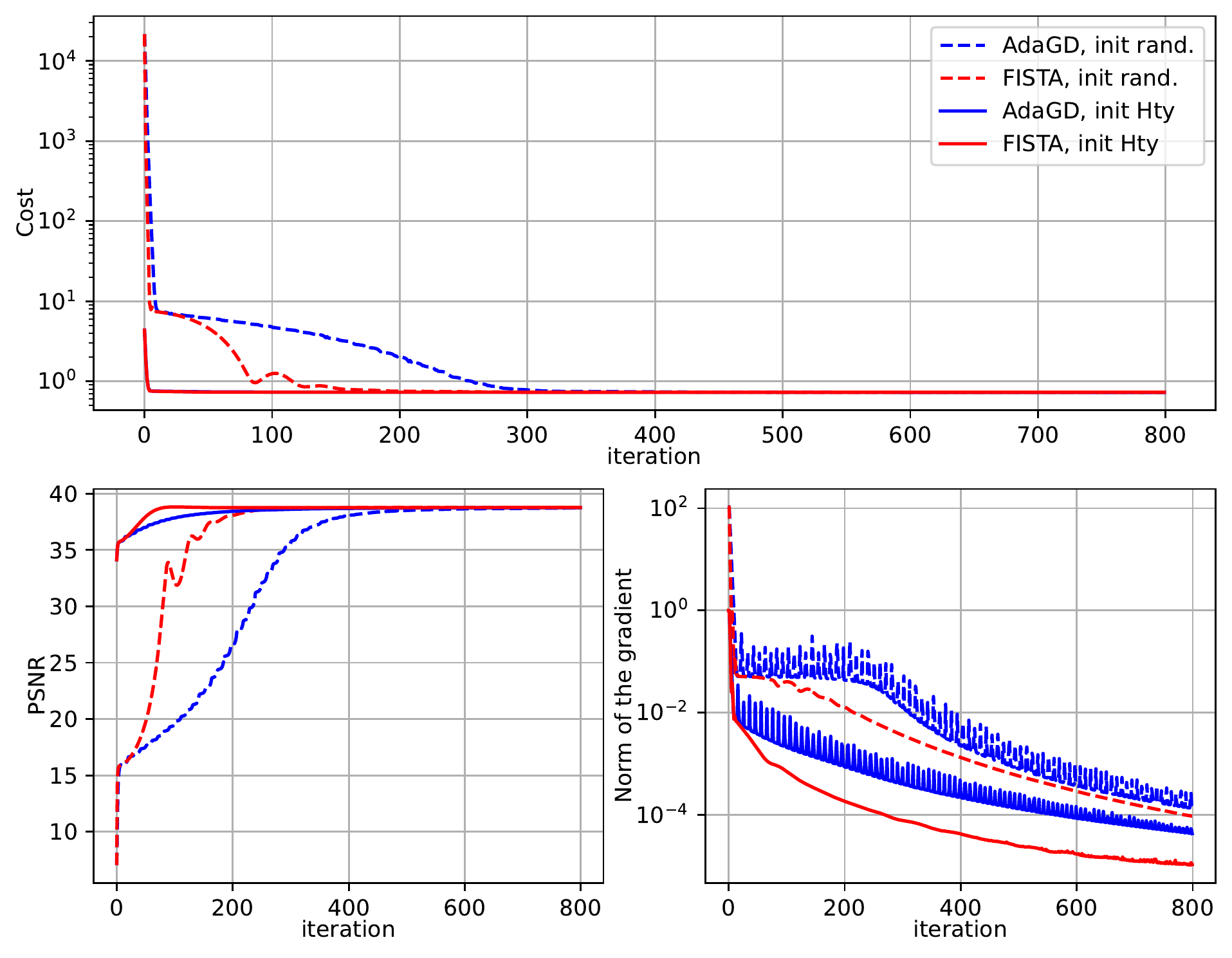}
        \caption{Example of convergence curves (MRI).}
    \label{fig:convergencemri}
\end{figure}
\renewcommand*{\thefootnote}{\fnsymbol{footnote}}
\renewcommand*{\thefootnote}{\arabic{footnote}}
\subsection*{Convergence Curves}
In this section, we present convergence curves for image denoising (Figure~\ref{fig:convergencedenoising}), MRI reconstruction (Figure~\ref{fig:convergencemri}), and CT reconstruction (Figure~\ref{fig:convergencect}) with CRR-NNs.
The underlying objective is minimized with FISTA\footnote{For the plots, the positivity constraint is dropped, otherwise, the gradient does not necessarily vanish at the minimum\label{note1}.}\footnote{For denoising, the problem is 1-strongly convex.
Hence, we use Nesterov's rule $(1 - \sqrt{L})/(1+\sqrt{L})$ instead of $(t_{k}-1)/t_{k+1}$ for extrapolation \cite{Nesterov2004}.}\cite{beck2009fast} and AdaGD\footref{note1}\cite{malitsky2020a}, which both converge generally fast.
Depending on the task and the desired accuracy, one or the other might be faster.
The observed gradient-norm oscillations for AdaGD are typical for this method and unrelated to CRR-NNs \cite{malitsky2020a}.
Finally, note that the initialization affects the convergence speed, but does not impact the reconstruction quality.
This differs significantly from PnP methods that deploy loosely constrained denoisers.

\begin{figure}[!ht]
    \centering
    \includegraphics[width=0.41\textwidth]{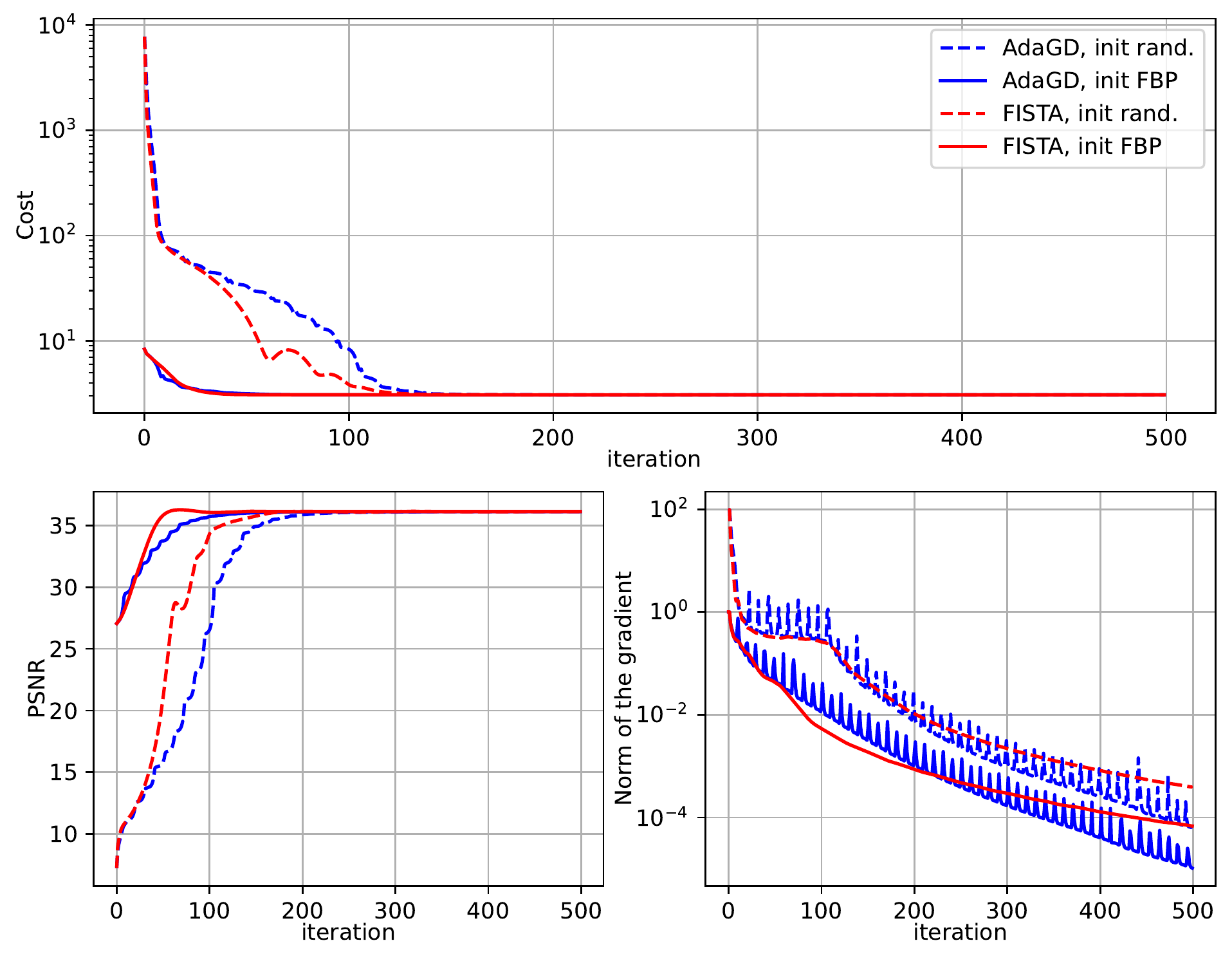}
        \caption{Example of convergence curves (CT).}
    \label{fig:convergencect}
\end{figure}
\subsection*{Activations and Filters}
We provide the filters and activations of a CRR-NN trained for the denoising of CT images (Figure~\ref{fig:filteractivationct}) and of MRI images (Figure~\ref{fig:filteractivationmri}).
Compared to the training on the BSD500 dataset, larger kernel sizes were needed to saturate the performances.

\begin{figure*}[t]
    \centering
    \begin{minipage}[c]{0.48\textwidth}
    \includegraphics[width=\textwidth]{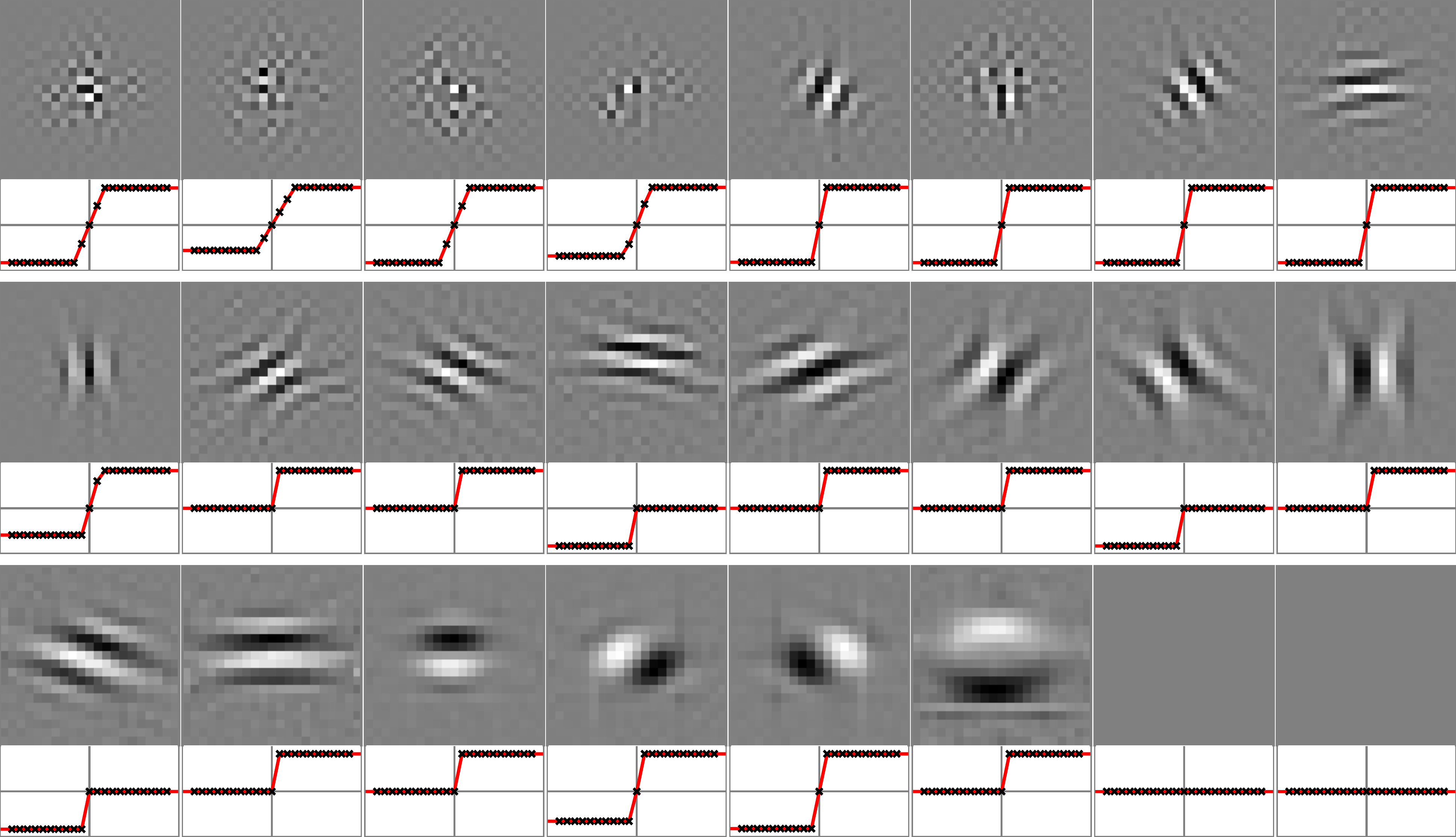}
    \caption{Impulse response of the filters and activation functions of the CRR-NN trained to denoise CT images.}
    \label{fig:filteractivationct}
    \end{minipage}
    \hspace{.2cm}
    \begin{minipage}[c]{0.48\textwidth}
    \includegraphics[width=\textwidth]{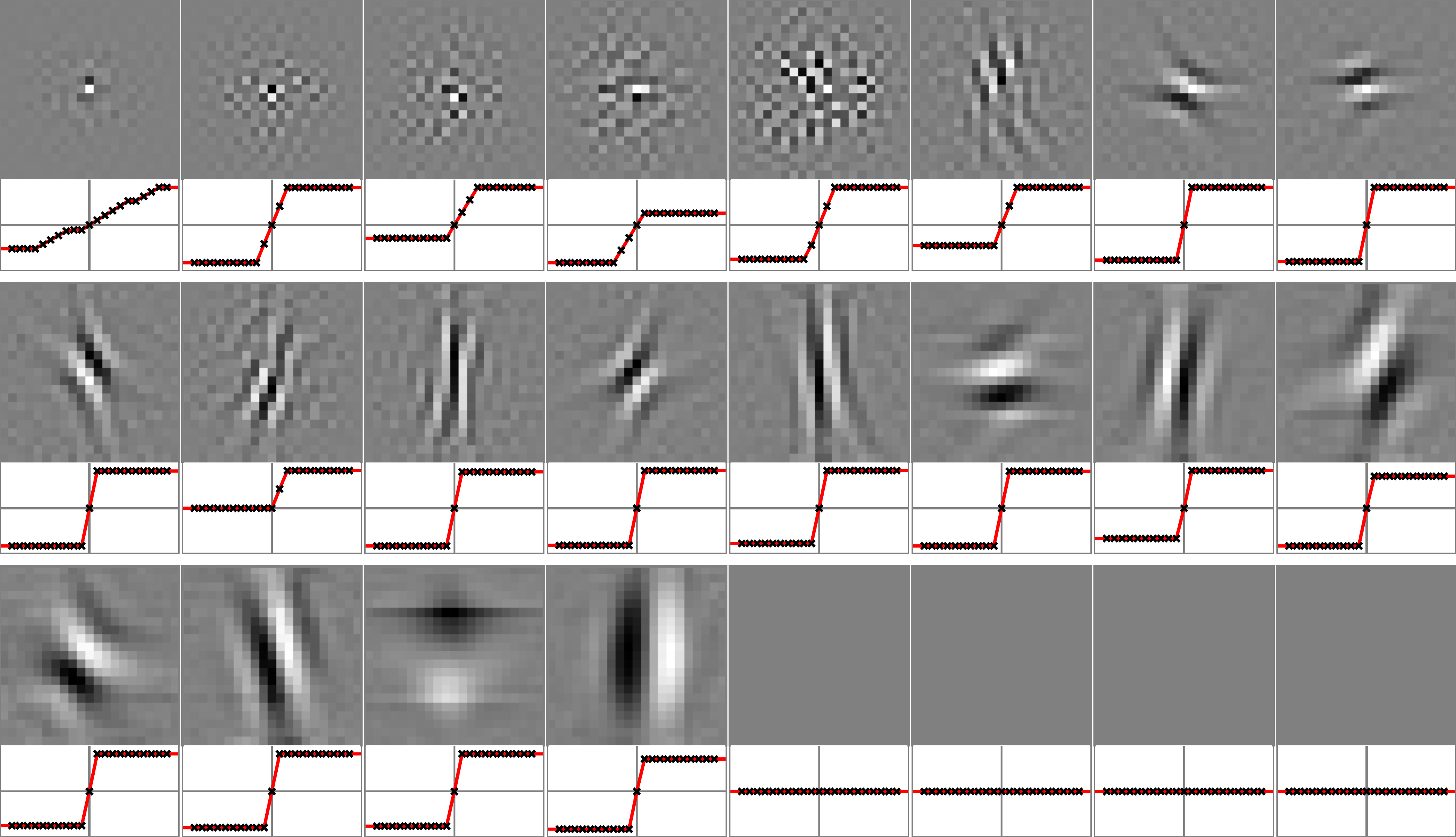}
    \caption{Impulse response of the filters and activation functions of the CRR-NN trained to denoise MRI images.}
    \label{fig:filteractivationmri}
    \end{minipage}
    \vspace{.8cm}
    
    \centering
    \includegraphics[width=0.95\textwidth]{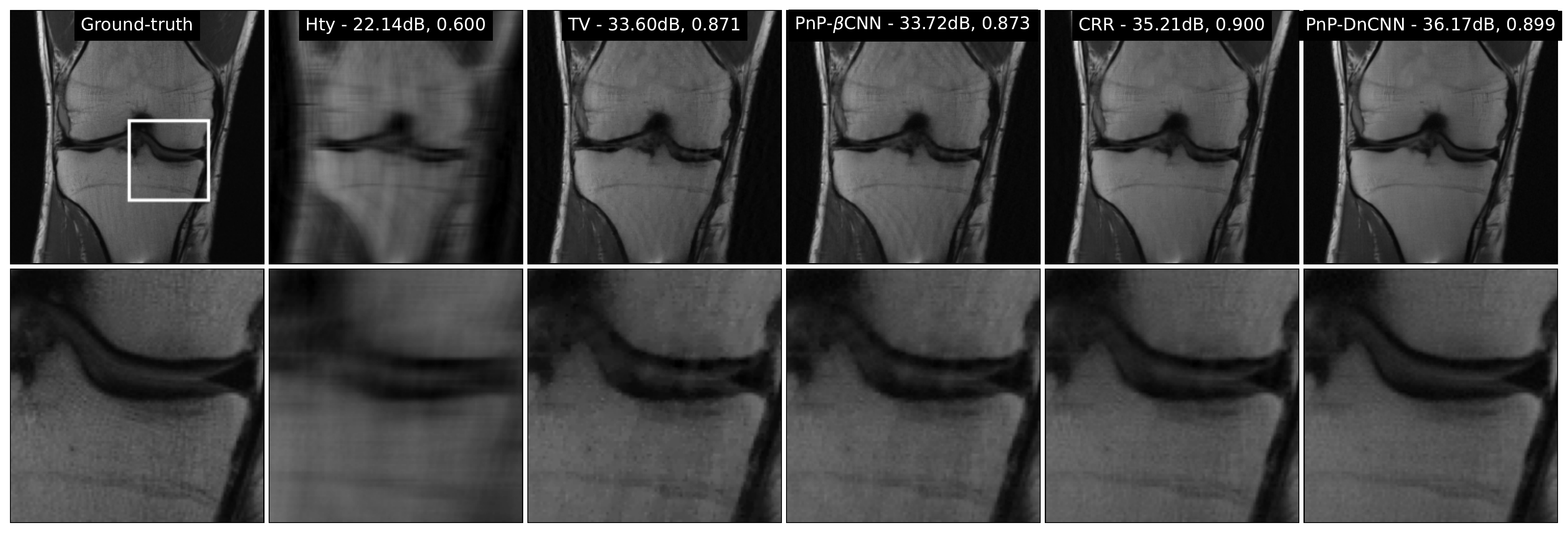}
    \caption{Reconstructions for the 8-fold accelerated multi-coil MRI experiment.}
    \label{fig:MRIreconstructionsextra1}
    \vspace{0.8cm}
    \centering
    \includegraphics[width=0.95\textwidth]{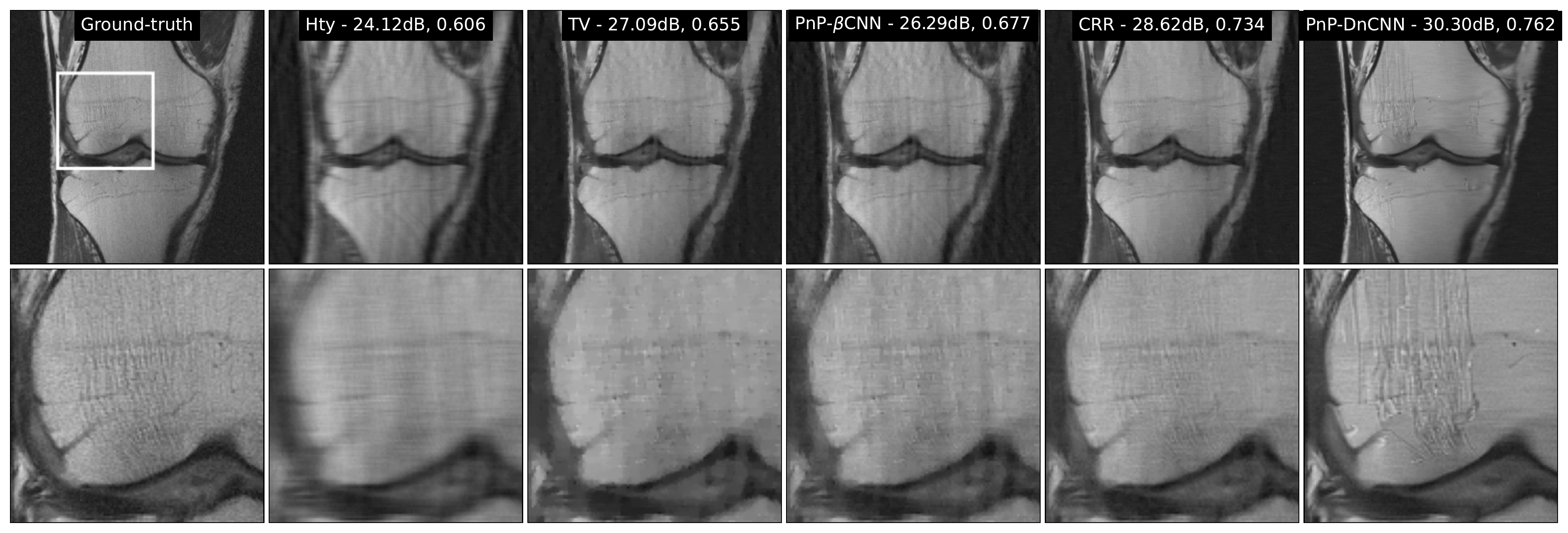}
    \caption{Reconstructions for the 4-fold accelerated single-coil MRI experiment.
    Note the unexpected behavior of DnCNN.}
    \label{fig:MRIreconstructionsextra2}
\end{figure*}
\subsection*{Reconstructed images}
\paragraph*{MRI}
In Figures \ref{fig:MRIreconstructionsextra1} and \ref{fig:MRIreconstructionsextra2}, we present reconstructions from multi- and single-coil MRI measurements, and report their PSNR and SSIM as metrics.
The reconstruction task in Figure~\ref{fig:MRIreconstructionsextra2} is particularly challenging.
In this regime, it can be observed that the loosely constrained PnP-DnCNN exaggerates some structures, even though the metrics remain acceptable.

\paragraph*{CT}
In Figures \ref{fig:CTreconstructionsextra1} and \ref{fig:CTreconstructionsextra2}, we provide reconstructions for the CT experiments with noise levels $\sigma_{\vec n} =1, 2$ in the measurements.
The reported metrics are PSNR and SSIM.
\clearpage
\begin{figure*}[t]    
    \centering
    \includegraphics[width=0.95\textwidth]{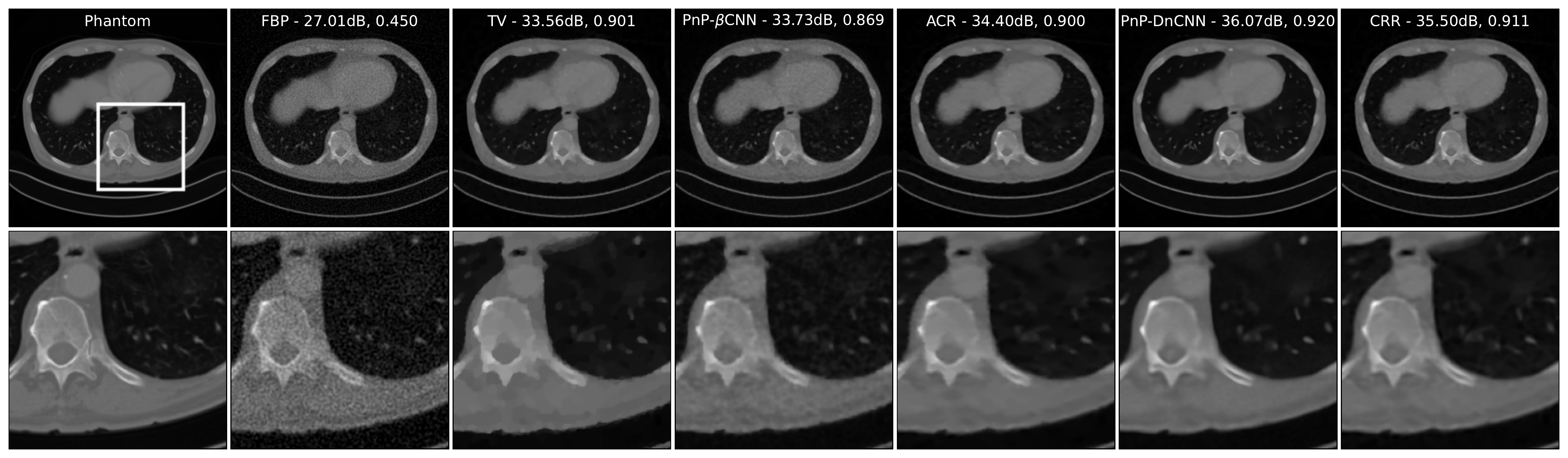}
    \caption{Reconstructed images for the CT experiment with $\sigma_{\vec n}=1.0$.}
    \label{fig:CTreconstructionsextra1}
    \vspace{0.8cm}
    
    \centering
    \includegraphics[width=0.95\textwidth]{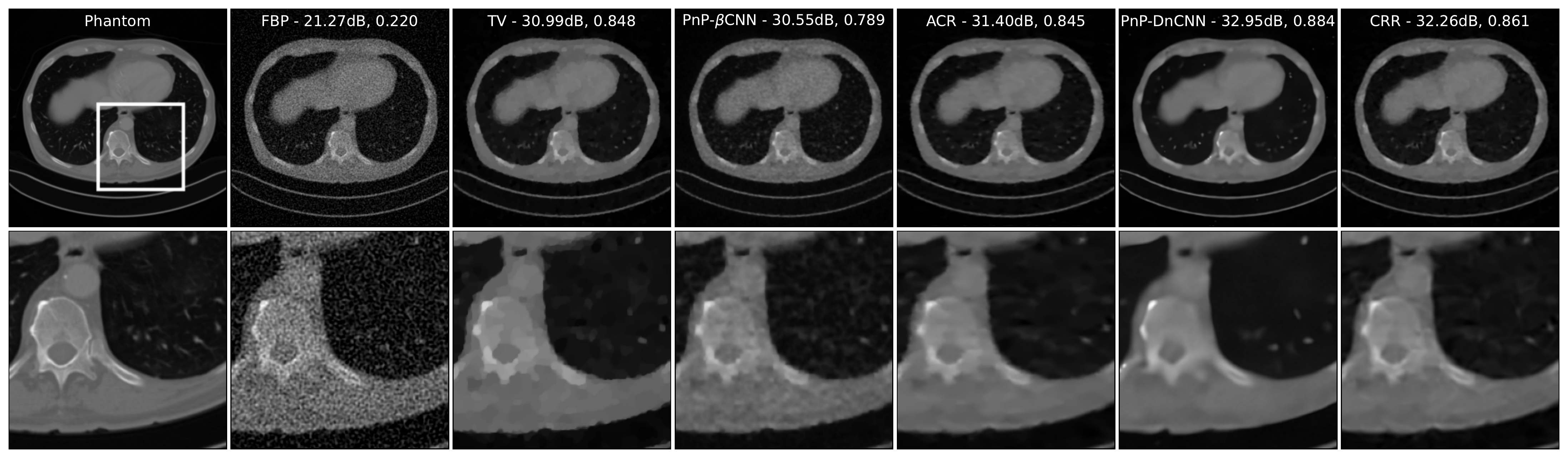}
    \caption{Reconstructed images for the CT experiment with $\sigma_{\vec n}=2.0$.}
    \label{fig:CTreconstructionsextra2}
\end{figure*}
\null
\end{document}